\documentclass[journal,twoside,web]{ieeecolor}
\usepackage{generic}
\renewcommand{\baselinestretch}{0.97}

\usepackage{comment}
\usepackage{microtype}
\usepackage{graphicx}
\usepackage{amsmath}
\usepackage{amssymb}
\usepackage{mathtools}

\usepackage{amsthm}
\usepackage{hyperref}
\usepackage[capitalize,noabbrev]{cleveref}
\usepackage{subcaption}
\theoremstyle{plain}
\newtheorem{theorem}{Theorem}[section]
\newtheorem{prop}[theorem]{Proposition}
\newtheorem{lemma}[theorem]{Lemma}
\newtheorem{corollary}[theorem]{Corollary}
\theoremstyle{definition}

\newtheorem{assu}[theorem]{Assumption}
\newtheorem{claim}[theorem]{Claim}
\newtheorem{fact}[theorem]{Fact}
\newtheoremstyle{boldremark}
    {\dimexpr\topsep/2\relax} 
    {\dimexpr\topsep/2\relax} 
    {}          
    {}          
    {\bfseries} 
    {.}         
    {.5em}      
    {}          

\theoremstyle{boldremark}
\newtheorem{brem}{Remark} 

\usepackage[textsize=tiny]{todonotes}
\usepackage{algorithm}
\usepackage{algorithmic}

\usepackage{setspace}

\DeclareMathAlphabet{\pazocal}{OMS}{zplm}{m}{n}

\newcommand\scalemath[2]{\scalebox{#1}{\mbox{\ensuremath{\displaystyle #2}}}}

\renewcommand\thetheorem{\arabic{section}.\arabic{theorem}}

\usepackage[utf8]{inputenc} 
\usepackage[T1]{fontenc}    

\usepackage{url}            
\DeclareMathAlphabet{\pazocal}{OMS}{zplm}{m}{n}

\usepackage{amsfonts}       
\usepackage{nicefrac}       
\usepackage{microtype}      
\usepackage{xcolor}         

\newcommand{\boundellipse}[3]
{(#1) ellipse (#2 and #3)
}

\usepackage{multirow}
\usepackage{multicol}
\usepackage{bm, bbm}
\usepackage{color}

\newcommand{\beq}{\begin{equation}} \newcommand{\eeq}{\end{equation}}
\renewcommand{\subsubsection}[1]{\noindent {\bf #1. }}  

\newcommand{\x}{\theta}
\newcommand{\g}{\bm{g}}
\newcommand{\G}{\bm{G}}
\newcommand{\X}{\Theta}

\renewcommand{\S}{\mathcal{S}}

\newcommand{\dg}{{k}_g}

\newcommand{\e}{\bm{e}}
\newcommand{\y}{\bm{y}}
\newcommand{\w}{\bm{w}}
\newcommand{\U}{{\bm{U}}}
\newcommand{\R}{{\bm{R}}}
\newcommand{\M}{\bm{M}}

\renewcommand{\a}{\bm{x}}
\newcommand{\ik}{{ti}}

\newcommand{\I}{\bm{I}}

\newcommand{\Z}{{\bm{{Z}}}}
\newcommand{\D}{{\bm{D}}}
\newcommand{\DD}{{\bm{F}}}
\newcommand{\Q}{{\bm{Q}}}

\newcommand{\SE}{\mathrm{SD}}  

\newcommand{\A}{\bm{X}}

\newcommand{\sg}{{\scalebox{.6}{(g)}}}
\newcommand{\ssg}{{\scalebox{.6}{g}}}
\newcommand{\sgp}{{{{\scalebox{.6}{{(}}}}{{\scalebox{.6}{{g}}}}'{{\scalebox{.6}{{)}}}}}}
\newcommand{\ssgp}{{{{\scalebox{.6}{{}}}}{{\scalebox{.6}{{g}}}}'{{\scalebox{.6}{{}}}}}}

\newcommand{\sone}{{\scalebox{.6}{{(1)}}}}

\newcommand{\ssone}{{\scalebox{.6}{{1}}}}
\newcommand{\sstwo}{{\scalebox{.6}{{2}}}}
\newcommand{\ssL}{{\scalebox{.6}{{L}}}}

\newcommand{\con}{{\footnotesize\mathrm{con}}}
\newcommand{\gd}{{\footnotesize\mathrm{gd}}}
\newcommand{\epsfin}{\eps}
\newcommand{\sj}{{\scalebox{.6}{{(j)}}}}
\newcommand{\gradU}{\mathrm{gradU}}
\newcommand{\hatgradU}{\widehat{\mathrm{gradU}}}
\newcommand{\conserr}{\mathrm{ConsErr}}
\newcommand{\err}{\mathrm{Err}}
\newcommand{\Uerr}{\mathrm{UconsErr}}

\newcommand{\svdeq}{\overset{\mathrm{SVD}}=} 
\newcommand{\qreq}{\overset{\mathrm{QR}}=} 
\newcommand{\bi}{\begin{itemize}} \newcommand{\ei}{\end{itemize}}
\newcommand{\ben}{\begin{enumerate}} \newcommand{\een}{\end{enumerate}}

\newcommand{\cblue}{\color{black}}

\renewcommand{\a}{\bm{x}}

\renewcommand{\b}{\bm{b}}

\newcommand{\B}{\bm{B}}
\newcommand{\V}{{\bm{V}}}
\newcommand{\W}{{\bm{W}}}

\newcommand{\z}{\bm{z}}

\newcommand{\indic}{\mathbbm{1}}

\newcommand{\C}{\bm{C}}

\newcommand{\tC}{\tilde{C}}

\newcommand{\Bstar}{{\B^\star}}   
\newcommand{\bstar}{\b^\star}             

\newcommand{\tb}{\rho}

\newcommand{\Xhat}{\hat\X}

\newcommand{\bhat}{\hat\b}

\newcommand{\Utilde}{\widetilde\U}

\newcommand{\init}{{\mathrm{init}}}

\newcommand{\Ustar}{\U^\star{}}
\newcommand{\Xstar}{\X^\star{}}
\newcommand{\xstar}{\x^\star}

\newcommand{\Vstar}{\V^\star{}}

\newcommand{\bSigma}{{\bm\Sigma^\star}}

\newcommand{\sigmin}{{\sigma_{\min}^\star}}
\newcommand{\sigmax}{{\sigma_{\max}^\star}}

\renewcommand{\P}{\bm{P}_{\footnotesize{\Ustar_\perp}}}

\newcommand{\E}{\mathbb{E}}
\newcommand{\norm}[1]{\left\|#1\right\|}

\newcommand{\pX}{\pazocal{X}}
\newcommand{\pY}{\pazocal{Y}}
\renewcommand{\V}{\bm{V}}

\newcommand{\bea}{\begin{eqnarray}}
\newcommand{\eea}{\end{eqnarray}}

\newcommand{\nn}{\nonumber}

\renewcommand{\forall}{\footnotesize\text{ for all }}

\newcommand{\eps}{\epsilon}
\newcommand{\ev}{\mathcal{E}}

\renewcommand{\bhat}{\b}  
  \renewcommand{\Xhat}{\X}

\newcommand{\sigmamin}{\sigma_{\min}}

\newcommand{\inperr}{\mathrm{InpErr}}

\newcommand{\epsconscalar}{\eps_{\con,sc}}

\newcommand{\trnc}{\mathrm{trnc}}
\newcommand{\inp}{\mathrm{in}}

\newcommand{\calG}{\mathcal{G}}
\newcommand{\calV}{\mathcal{V}}
\newcommand{\calE}{\mathcal{E}}
\newcommand{\N}{\mathcal{N}}
\newcommand{\out}{{\text{out}}}

\newcommand{\true}{\text{true}}
\newcommand{\bphi}{\bm{\Phi}}

\newcommand{\Upm}{\U}
\newcommand{\Rpm}{\R}

\newcommand{\Urand}{\tilde{\Upm}_\init}
\newcommand{\Urandorth}{\Upm_\init}

\renewcommand{\t}{{k}}

\newcommand{\Tconexp}{{ \frac{1}{\log(1/\gamma(\W))}}}

\begin{document}

\title{Decentralized  Communication-Efficient  Multi-Task Representation Learning}

\author{Shana Moothedath,~\IEEEmembership{Member,~IEEE}, Namrata Vaswani,~\IEEEmembership{Fellow,~IEEE,}
\thanks{S. Moothedath and N. Vaswani are with Electrical and Computer Engineering, Iowa State University. Email: $\{$mshana, namrata$\}$@iastate.edu.}
\thanks{This work is supported by NSF grant 2213069.}}

\maketitle

\begin{abstract}
\cblue
Representation learning enables effective learning in data-scarce settings by extracting shared features from related tasks. While well-studied in the centralized setting, decentralized representation learning remains underexplored. In decentralized settings, data is distributed across nodes that must collaboratively learn without a central coordinator.
In this work, we study decentralized multi-task learning under a shared low-dimensional representation.  We consider $T$ source tasks each with $n$ data points in $\mathbb{R}^d$. The goal is to recover the matrix $\Xstar:= [\xstar_1, \xstar_2, \ldots, \xstar_T] \in \mathbb{R}^{d \times T}$, which has rank $r \ll \min\{d,T\}$, from observations following the model $\y_\ik:= \a_\ik^\top \xstar_t,$ for $t=1,2,...T,$ and $i=1,2,\ldots, n.$
We present a novel alternating projected gradient descent and minimization algorithm for recovering a low-rank feature matrix in a decentralized fashion.
%
 We obtain constructive provable guarantees that provide a lower bound on the required sample complexity and an upper bound on the iteration complexity (total number of iterations needed to achieve a certain error level) of our proposed algorithm. This latter bound allows us to analyze the time and communication complexity of our algorithm and show that it is fast and communication-efficient.
%
 We performed numerical simulations to validate the performance of our algorithm and compared it with benchmarks.

\end{abstract}

\section{Introduction}
{\cblue Multi-task representation learning aims to learn a shared representation across related tasks, enabling models to leverage common structure and improve overall performance \cite{wang2016distributed}.}
By sharing knowledge across tasks, multi-task learning can lead to more effective models, especially when data is limited or expensive. 
{\cblue This paradigm has demonstrated remarkable success in various applications, such as natural language processing (GPT-2, GPT-3, Bert), vision domains, and robotics \cite{caruana1997multitask,bengio2013representation}, though its theoretical foundations remain underexplored.
Specifically, decentralized methods with provable guarantees remain largely underexplored.}
A prevalent assumption in the literature is the presence of a shared common representation among different tasks \cite{du2020few, tripuraneni2021provable, collins2021exploiting}. 
In this paper, we study decentralized multi-task learning when multiple tasks share a common set of low-dimensional features. In the decentralized setting, there is no central coordinating node and each task can only exchange information with a subset of tasks defined by a communication network.  Our aim is the following: 
{\em Given a set of diverse samples from $T$ different tasks that are distributed in a decentralized manner, how can we efficiently (and optimally) learn a common feature representation?}

Decentralized multi-task learning proves particularly useful in many applications where diverse tasks can be simultaneously addressed across decentralized entities in a privacy-preserved manner.
Decentralized learning provides a much faster and more scalable learning paradigm than its traditional counterpart (all nodes communicating with the center) in settings where the nodes are geographically distant from the center and the number of tasks is large. 
 Additionally, relying on a central server introduces the risk of a potential single point of failure.
In this work, we develop and analyze a decentralized gradient descent (GD) based algorithm for solving the multi-task problem.

%
%
\subsection{Problem Setting and Notations}
\subsubsection{Problem setting}
Consider $T$ tasks. Each task $t \in [T]$ is associated with a distribution $\mu_t$ over a joint data space $\pX \times \pY$, where $\pX$ is the input space and $\pY$ is the output space. We consider $\pX \subseteq \mathbb{R}^d$ and $\pY \subseteq \mathbb{R}$. For each task $t \in [T]$ we have access to $n$ i.i.d samples $(\a_\ik, \y_\ik)$ from $\mu_t$,
where the data follows the statistical model
\[\y_\ik:= \a_\ik^\top \xstar_t, \quad t=1,2,...T,\quad i=1,2,\ldots, n.\]
We represent the $n$ samples as an input matrix $\A_t \in \mathbb{R}^{n \times d}$ and an output vector $\y_t \in \mathbb{R}^n$.
Multi-task learning aims to simultaneously learn prediction functions for all $T$ tasks, to uncover an underlying shared property among these tasks, referred to as {\em representation}. Specifically, we consider a low-dimensional representation, where for each task $t \in [T]$, there is an unknown vector $\xstar_t \in \mathbb{R}^d$ and
\beq
\y_t:= \A_t \xstar_t, ~t=1,2,\ldots, T.
\label{ykvec}
\eeq
The goal is to recover a set of $T$ $d$-dimensional feature vectors, $\xstar_1, \xstar_2, \ldots, \xstar_T$ that are such that the $d \times T$ matrix  $\Xstar:= [\xstar_1, \xstar_2, \ldots, \xstar_T]$ has rank $r \ll \min\{d,T\}$.  
The regime of interest is $n < d$ since we consider data-scarcity, and the goal is to have to use as few samples $n$ as possible. 

In this paper, we consider a decentralized setting: there are a total of $L \le T$ distributed nodes/agents, and each node has access to a different subset of the task data $(\A_t, \y_t)$s.
We denote the set of indices of the $\y_t$s available at node $g$  by $\S_g$. The sets $\S_g$ are mutually disjoint and $\cup_{g=1}^L \S_g = [T]$. Here $[T]:=\{1,2,\dots, T\}$.
The $L$ nodes are connected by a communication network whose topology (connectivity) is specified by an undirected graph $\calG= (\calV, \calE)$, where $\calV$ is the set of $L$ nodes ($|\calV|=L$) and $\calE$ denotes the set of edges. The neighbor set of the $g$-th node is denoted by $\N_g (\calG)$, i.e., $\N_g (\calG) := \{j: (g, j) \in \calE\}$. 
Since there is no central coordinating node, each node can only exchange information with its neighbors.
No node can recover the entire matrix $\Xstar$. Node $g$ only recovers the feature vectors $\xstar_t$, for $t \in \S_g$, i.e., the sub-matrix $\Xstar^\sg:= [\xstar_t, \ t \in \S_g]$.

\subsubsection{Preliminaries}
Let us denote its reduced (rank $r$) SVD of $\Xstar$ , $\Xstar \svdeq \Ustar {\bSigma \V^{\star}} := \Ustar  \B^\star$,
where $\Ustar$ and ${\V^\star}^\top$ are matrices with orthonormal columns {\em (basis matrices)}, $\Ustar$ is  $d \times r$, $\V^\star$ is $r \times T$, and $\bSigma$ is an $r \times r$ diagonal matrix with non-negative entries (singular values). We let  $\B^\star:= \bSigma \V^{\star}=[\bstar_1,\ldots, \bstar_T]$.
We use $\sigmax$ and $\sigmin$ to denote the maximum and minimum singular values of $\bSigma$ and we define its condition number  as
 $\kappa:= \sigmax/\sigmin.$
%
%
%
%
Note that $\y_t$s are not global functions of $\Xstar$: no $\y_\ik$ is a function of the entire matrix $\Xstar$. 
We thus need an assumption that enables correct interpolation across the different columns. The following incoherence (w.r.t. the canonical basis) assumption on the right singular vectors suffices for this purpose.
 Such an assumption on both left and right singular vectors was first introduced in \cite{matcomp_candes} for making the low-rank matrix completion problem well-posed and used in recent works on multi-task representation learning \cite{tripuraneni2021provable, collins2021exploiting}. 

\begin{assu}[Incoherence of right singular vectors] \label{right_incoh2}
Assume that $\|\bstar_t\|^2 \le \mu^2 r \sigmax^2 / T$, $t\in [T],$ for a numerical constant $\mu$ that is greater than one. Since $\|\xstar_t\|^2 = \|\bstar_t\|^2$, this also implies  that $\|\xstar_t\|^2 \le \mu^2 r \sigmax^2 / T$.
Here, and below, $\|.\|$ without a subscript denotes the $\ell_2$ norm.
\end{assu}
We also have the two commonly used assumptions below. 
\begin{assu}
The matrices $\A_t$ are independent and identically distributed (i.i.d.) random Gaussian (each matrix entry is i.i.d. standard Gaussian).
\label{iidAk}
\end{assu}
\begin{assu}
The graph, $\calG$, specifying the network topology {\cblue is undirected and} connected. {\cblue That is,  there exists a path between any two nodes in the graph}. 
\label{connectedcalG}
\end{assu}
{\cblue Assumption~\ref{connectedcalG} is standard in decentralized control and learning literature \cite{olshevsky2009convergence, nedic2009distributed, xin2021fast}. The Gaussian model in Assumption~\ref{iidAk} is commonly assumed in most works on MTRL with theoretical guarantees \cite{tripuraneni2021provable, collins2021exploiting, OurICML}, and potential extensions beyond this assumption are part of our future work.}

\subsubsection{Notations}
We use $[T]:=\{1,2, \dots, T\}$.
We denote the Frobenius norm as $\norm{\cdot}_F$, and the induced $\ell_2$ norm as $\norm{\cdot}$ without any subscript. 
Further, ${}^\top$ denotes matrix or vector transpose. For a tall matrix $\M$, $\M^\dagger:= (\M^\top \M)^{-1} \M^\top$ (pseudo-inverse).
%
%
At various places in the paper, for tall $d \times r$ matrices, we are only interested in the orthonormal basis spanned by the matrix. We compute this basis using QR decomposition, i.e. $\tilde\U \qreq \U \R$. We refer to the matrix $\U$ as the {\em basis matrix}.
%
For two $d \times r$ basis matrices $\U_1, \U_2$ (corresponding to two subspaces defined by their column spans), the measure of Subspace Distance (SD) used is
$
\SE_2(\U_1, \U_2) : = \|(\I  - \U_1 \U_1^\top) \U_2\|.
$
%
 We reuse $c,C$ to denote different constants in each use with $c < 1$ and $C \ge  1$ (these do not depend on the model parameters $d,T, r$). This is a standard approach in literature when only the sample order or time complexity is of interest.
 The notation $a \gtrsim b$ means that $a \ge C b$.
We will prove our results to hold with probability (w.p.) at least $1- 1/n$ or $1-C/n$ for a small constant $C$. We say a bound holds with high probability (whp) if it holds with at least this much probability.
%
%
%
\subsection{Contributions}
 We develop a decentralized alternating GD and minimization-based approach for solving the decentralized multi-task representation learning (Dec-MTRL) problem.  We provide a {\em non-asymptotic theoretical guarantee} for our proposed algorithm.
 Our theoretical guarantee provides an order-wise lower bound on the required sample complexity and an upper bound on the required number of iterations to achieve a certain accuracy level (iteration complexity). This upper bound enables us to also obtain time and communication complexity guarantees, and argue that our approach is both fast and communication-efficient. 
{\cblue Low-dimensional representation learning has been widely studied in recent works \cite{du2020few, tripuraneni2021provable, OurICML}, but decentralized approaches with provable guarantees remain largely underexplored. To the best of our knowledge, no existing work provides a provably accurate algorithm for solving the Dec-MTRL problem. Although there is a substantial body of work on decentralized optimization, most focus on convex problems or offer only asymptotic guarantees. Some recent efforts address non-convex settings, but they rely on strong additional assumptions. As a result, existing algorithmic or proof techniques cannot be directly applied to our setting.} 
 Simulation experiments are used to corroborate our results and compare them with other benchmark approaches.
\vspace{-3 mm}
\section{Related Work}\label{sec:rel}
\subsubsection{Multi-task representation learning}
Representation learning has shown its great power in various
domains  \cite{bengio2013representation}. 
The first theoretical analysis with sample complexity bounds was provided in \cite{baxter2000model}.
\cite{maurer2016benefit} and follow-up works considered the setting where all tasks are i.i.d sampled from a certain distribution and analyzed the benefits of representation learning for reducing the sample complexity of the target task.
Building upon their findings, \cite{du2020few} delved into few-shot learning through representation learning.
\cite{tripuraneni2021provable} specifically addressed the challenge of multi-task linear regression with low-rank representation, presenting algorithms with robust statistical rates.
 \cite{tripuraneni2021provable}  also gives a computationally efficient algorithm for standard Gaussian inputs and a lower bound for subspace recovery in the low-dimensional linear setting. 
 \cite{collins2021exploiting} proposed an alternating minimization and gradient decent approach.
 All of these works considered a centralized setting, where a central server orchestrates the joint learning.
 
 On the other hand, somewhat similar problems have received some attention in the recent low-rank matrix recovery literature motivated by applications such as federated sketching and dynamic MRI \cite{lrpr_gdmin,lrpr_gdmin_2}.
A convex relaxation (mixed norm minimization) approach in the centralized
setting was proposed in \cite{lee2019neurips}.
  A fast GD centralized solution was studied in \cite{collins2021exploiting, lrpr_gdmin,lrpr_gdmin_2}.
Recently, \cite{moothedath2022fully,moothedath2022dec, moothedath2023comparing, moothedath2024decentralized} studied the decentralized low-rank matrix recovery problem. 
\cite{moothedath2022fully, moothedath2022dec} proposed a high-level algorithmic approach.
However, the algorithm required significant modifications to achieve provable guarantees. 
Later the conference versions of this paper \cite{moothedath2023comparing, moothedath2024decentralized} proposed an updated algorithm; however,  no theoretical proofs were provided.
Building on \cite{moothedath2023comparing, moothedath2024decentralized}, this paper establishes theoretical guarantees and provides rigorous proofs for the decentralized MTRL problem.
\subsubsection{Consensus-based (gossip) algorithms}
 Since the communication model we consider in this work is decentralized,  henceforth we focus on decentralized algorithms. 
There is a large amount of literature on using multiple consensus iterations to approximate the gradient (or related quantities such as the Hessian weighted gradient used in Newton Raphson) \cite{li2013convergence, wai2017decentralized} for both unconstrained and constrained optimization problems. However, there are a few differences. {\cblue Most of these works study a general convex optimization problem and hence can only prove convergence to {\em a} minimizer. When non-convex problems are studied, then only convergence to a stationary point can be proved.} Also, almost all guarantees provided are asymptotic and require many assumptions. 

A second difference is that all works consider decentralized versions of optimization problems of the form
$
\min_\theta  \mathcal{C} \sum_{g=1}^L f_g(\theta).
$
In contrast, our work requires solving the problem:
$
\min_{\U,\B: \U^\top \U= \I} \sum_{g=1}^L  [ \sum_{t \in \S_g} ||\y_t - \A_t \U \b_t||_2^2].$ 
In our problem, $\U$ is a global variable, while $\B$ is not: $\B = [\B_1, \B_2, \dots, \B_g, \dots, \B_L]$ and its sub-matrices $\B_g$ only occur in $f_g(.)$. Thus these can be updated locally at respective nodes $g$. This simplifies our algorithm (only gradients w.r.t. $\U$ need to be shared), but this also means that the proof requires more work. 
Also, our goal is to recover the unknown LR matrix $\Xstar$; consequently we need our developed algorithm to converge to this specific minimizer. For this, (i) our algorithm needs a careful decentralized initialization; and (ii) our guarantee needs a lower bound on the required sample complexity. 

Lastly, even though our optimization problem is an instance of nonlinear LS (NLLS), it is not easy for us to use the standard optimization approach for NLLS: this involves obtaining a closed form expression for $\B$s: $\b_t = (\A_t \U)^\dag \y_t$ and substituting this into the cost function to get a new cost function that depends only on $\U$. One then develops a GD or Newton Raphson algorithm to minimize the new cost which in our case will be
$
\min_{\U,\B: \U^\top \U= \I} \sum_{g=1}^L [ \sum_{t \in \S_g} ||[\I - \A_t \U (\A_t \U)^\dag ] \y_t||_2^2 ]
$
over only $\U$. However, it is not clear how to compute the gradient w.r.t $\U$ of the above in closed form. Even if one could, analyzing the resulting algorithm would be extremely difficult. Consequently we cannot borrow the decentralizing approach developed for NLLS in \cite{li2013convergence}.

\subsubsection{Decentralized non-convex optimization}
\cite{xin2021fast} studied a randomized incremental gradient method for decentralized unconstrained minimization of a non-convex cost function that is smooth and that satisfies the Polyak \L ojasiewicz (PL) condition, which is a generalization of strong convexity to non-convex functions (if a function satisfies PL condition, then all its stationary points are global minimizers). 
The result in \cite{xin2021fast} also shows that this number grows linearly with $1/\eps$.
%
\section{Decentralized Representation Learning}\label{sec:alg}
For the $T$ tasks, the joint optimization problem that needs to be solved  can be stated as

$
\min_{\X: \X \text{ is rank } r} \tilde{f}(\X): = \sum_{t=1}^T \|\y_t - \A_t \x_t\|^2.
$

We rewrite the unknown matrix $\X$ as $\X= \U \B$, where $\U$ is $d \times r$ and $\B$ is $r \times T$, and consider the cost function
\begin{align}
\scalemath{0.93}{f(\U, \B)=\hspace{-2 mm}\sum_{g=1}^L f_g(\U, \B)=\hspace{-2 mm}\sum_{g=1}^L \sum_{t \in \S_g} \|\y_t - \A_t \U \b_t\|^2.}\label{dec_cost_fn}
\end{align}

Notice that Eq.~\eqref{dec_cost_fn} is non-convex. {\cblue Unlike convex
problems, where zero or random initialization may suffice, non-convex optimization requires more deliberate strategies
for initialization to ensure convergence. Hence we require careful initialization.} Our proposed Decentralized Low-Rank Representation Learning (Dec-MTRL) algorithm consists of three stages: (i)~A careful initialization-spectral initialization, (ii)~Minimization over $\B$, and~(iii)~Projected-GD (ProjGD) for $\U$. 
%
%
Due to the non-convexity of the cost function $f(\U,\B)$ and the constraint on $\U$, we are unable to borrow decentralized optimization and related ideas. We use a scalar average-consensus algorithm on each entry of the gradient at each node to approximate its sum. 
%
\subsection{Average Consensus (AvgCons) algorithm}
\label{cons_algo_sec}
\begin{algorithm}[t]
\caption{AvgCons: Average consensus} 
\label{alg:consensus}
\label{avgcons_algo}
\begin{algorithmic}
\STATE \textit {\bf Input:} $\Z_\inp^\sg$, for all $g \in [L]$  ($\Z_\inp^\sg$ is a matrix), weight matrix $\W \in \mathbb{R}^{L \times L}$, $T_\con$, $\calG$ (graph connectivity)
\end{algorithmic}
\begin{algorithmic}[1]
\STATE Initialize $\Z_0^\sg \leftarrow \Z_\inp^\sg$, for all $g \in [L]$
\FOR{$\tau=1$ to $T_\con$}
\STATE $\Z^\sg_{\tau+1}  \leftarrow \Z^\sg_\tau +  \displaystyle\sum_{j \in \N_g (\calG)} \W_{gj} \Big( \Z^\sj_\tau- \Z^\sg_\tau \Big)$, for $g \in [L]$\label{step:avg}
\ENDFOR

\STATE \textit {\bf Output:}  $\Z_\out^\sg \leftarrow L \cdot \Z^\sg_{T_\con}$
\end{algorithmic}
\end{algorithm}
We summarize the average consensus algorithm in Algorithm \ref{alg:consensus}.
In each of its iterations, nodes update their values by taking a weighted sum of their own and their neighbors' partial sums. With enough iterations, if the graph is connected, one can approximate the true sum at each node \cite{olshevsky2009convergence}.
The weight matrix $\W$ is symmetric and doubly stochastic\footnote{Since $\W$ is doubly stochastic, its largest eigenvalue is $1$ and all eigenvalues are between $1$ and $-1$.} 
We use $\lambda_k(\W)$ to refer to the $k$-th largest eigenvalue of $\W \in \mathbb{R}^{L \times L}$. Then the second largest magnitude eigenvalue of $\W$ is defined as
$\gamma(\W) := \max(|\lambda_2(\W)|,|\lambda_L(\W)|).$
The result below presents a convergence result for Algorithm \ref{alg:consensus}.
%
%
%
\begin{prop}[\cite{olshevsky2009convergence}.] \label{avgcons_prop_scalar} 
Consider the average consensus algorithm (AvgCons) given in Algorithm \ref{avgcons_algo} with a doubly stochastic weights' matrix $\W$.
Let $z_{true}:=\sum_{g=1}^L z_{\inp}^\sg$ be the true sum that we want to compute in a decentralized fashion.
Pick an $\eps_\con < 1$. If the graph of the network is connected, and if $T_\con \ge  \frac{1}{\log(1/\gamma(\W))} \log(L/\eps_\con)$, then
\[
\max_g | z_{out}^\sg - z_{true} | \le  \eps_\con \max_g | z_{\inp}^\sg - z_{true} |.
\]
\end{prop}
%
\subsection{Initialization for Dec-MTRL Estimator}
Since $f(\U,\B)$ is not a convex function of the unknowns $\{\U, \B\}$, a careful initialization is needed. We utilize the spectral initialization approach. 
For initialization, our goal is to compute the top $r$ left singular vectors of
\begin{align*}
\X_0^{central}:= & \left[ \frac{1}{n} (\A_1^\top \y_{1,\trnc}(\alpha)),  \dots,   \frac{1}{n} (\A_T^\top \y_{t,\trnc}(\alpha)) \right]
\label{newinit}
\end{align*}
where $\alpha := \tC \frac{\sum_\ik  (\y_\ik)^2}{nT}$ and $\y_{t,\trnc}(\alpha) := \y_t \circ \indic_{\{|\y_t| \le  \sqrt\alpha\}}$ and $\tC = 9 \kappa^2 \mu^2$. Using $\e_t$'s, this can be expressed as
$\X_0^{central} = (1/n) \sum_{t \in [T]} \sum_{t=1}^T \sum_{i=1}^n \a_\ik \y_\ik \e_t{}^\top \indic_{\left\{ \y_\ik^2 \le \alpha \right\} }$.
Here, $\e_t$ denotes the $t^{\rm th}$ canonical basis vector.
We use $\indic_{\mathrm{\{statement\}}}$ to denote an indicator function that takes the value $1$ if the statement is true and zero otherwise.
For a vector $\w$ and a scalar $\alpha$,  $\indic_{\{\w \le \alpha\}}$ returns a vector of 1s and 0s of the same size with 1s where $(\w(k) \le \alpha)$ and zero everywhere else. Also, $|\w|$ takes the element-wise magnitude of the vector and $\circ$ denotes the Hadamard product. Thus $\y_{t,\trnc}(\alpha) := \y_t \circ \indic_{\{|\y_t| \le  \sqrt\alpha\}}$ zeroes out entries of $\y_t$ with magnitude larger than $\alpha$. 

Here, observe that we are summing $\a_\ik \y_\ik \e_t{}^\top$ over only those values $i,t$ for which $\y_\ik^2$ is not too large (it is not much larger than its empirically computed average value). This {\em truncation} filters out the too large (outlier-like) measurements and sums over the rest. Theoretically, this converts the summands into sub-Gaussian r.v.s which have lighter tails than the un-truncated ones, which are sub-exponential.
The idea of ``truncating'' (zeroing out very large magnitude entries of the observation vector to make it lighter tailed) is borrowed from \cite{twf}.
%
The pseudocode for the the computation of the threshold $\alpha$ is given in Algorithm \ref{altgdmin_dec_init}. 
Each node $g$ truncates its $\y_t$s (zeros out entries with magnitude greater than $\sqrt{\alpha^\sg}$) to obtain $\y_{t,\trnc}$. This is then used to compute the $d \times |\S_g|$ sub-matrix of $\X_0$,
\[
(\X_0)_\ssg = [\frac{1}{n} \A_t^\top \y_{t,\trnc}(\alpha^\sg), \ t \in \S_g].
\]
The full matrix
$
\X_0 : =[(\X_0)_\ssone, (\X_0)_\sstwo, \dots, (\X_0)_\ssL ]
$
is not available at any one node. However, we do not need this full matrix; we only need its top $r$ singular vectors.

We use a consensus-based approximation of the power method (PM) \cite{hardt2014noisy} to approximate the desired singular vectors. 
Using the same random seed, each node $g$ generates the same $\Urandorth$ by generating a $d \times r$ matrix with i.i.d. standard Gaussian entries and orthonormalizing it. Since the same random seed is used, this matrix is the same at each node. This serves as the initialization for the PM, i.e., $\Upm^\sg = \Urandorth$.
This same random seed assumption is made in previous works that study decentralized PM, e.g., \cite{decpm_first, Waheed}.
Next, each node computes
$
 \tilde\Upm_\inp^\sg =  (\X_0)_\ssg  (\X_0)_\ssg{}^\top \Upm^\sg.
$

This is the input to the AvgCons algorithm, which outputs at node $g$,
${\tilde\Upm}^\sg = \text{AvgCons}_g(\tilde\Upm_\inp^\sgp, \ g' \in [L]).$
This is an approximation to $\sum_g (\X_0)_\ssg  (\X_0)_\ssg{}^\top \Upm^\sg $. 
Finally,  node-1 uses QR decomposition on ${\tilde\Upm}^\sone$ to get $\Upm^\sone$, i.e.,
${\tilde\Upm}^\sone \qreq \Upm^\sone \Rpm^\sone.$
Then we run a consensus loop to share $\Upm^{\sone}$ with all the nodes as follows: node-1 transmits $\Upm^{\sone}$, while all other nodes transmit a zeros' matrix. All nodes other than node-1 use the output of this consensus loop as their $\Upm^\sg$. At each iteration $\tau$, this step guarantees consensus of $\Upm^\sg$s to error level $\eps_\con$.
We summarize this in Algorithm \ref{altgdmin_dec_init}. We assume all nodes to know which node is ``node-1'', as in \cite{olshevsky2014linear}.

The final output of the initialization (Algorithm~\ref{altgdmin_dec_init}) serves as the input to Algorithm~\ref{altgdmin_dec}. We specify the complete Dec-MTRL algorithm in Algorithm~\ref{altgdmin_dec}.
This uses sample-splitting which is a commonly used approach in the LR recovery literature \cite{lrpr_gdmin, lowrank_altmin,rmc_gd, fastmc} to simplify the analysis. It helps ensure that the measurement matrices in each iteration for updating each of $\U$ and $\B$ are independent of all previous iterates. This allows one to use concentration bounds for sums of independent random variables to bound the error terms at each iteration.

\begin{algorithm}[t]
\caption{Initialization for Dec-MTRL} 
\label{altgdmin_dec_init}
\begin{algorithmic}[1]
\STATE Let $\y_t \equiv \y_t^{(00)}, \A_t \equiv \A_t^{(00)}$ for all $t \in [T]$.
\STATE $\alpha^\sg_\inp  \leftarrow  9\kappa^2 \mu^2 \frac{1}{nT} \displaystyle\sum_{t \in\S_g}\sum_{i=1}^n \y_{\ik}^2$ \label{step:cons-1} $\forall g \in [L]$
\STATE $\alpha^\sg \leftarrow  \mathrm{AvgCons}_g(\alpha^\sgp_\inp, g' \in [L], \calG, T_\con)$
\STATE Let $\y_t \equiv \y_t^{(0)}, \A_t \equiv \A_t^{(0)}$ 
\STATE $\y_{t,\trnc}:= \y_t \circ \indic\{|\y_\ik| \le \sqrt{\alpha^\sg} \}$ $\forall t \in \S_g$, $ g \in [L]$
\STATE Compute $\scalemath{0.95}{(\X_0)_\ssg =  \left[\frac{1}{n} \A_t^\top \y_{t,\trnc}, \ t \in \S_g \right]$, $\scalemath{0.9}{\forall g \in [L]}}$
\STATE Generate $\Urand$ with each entry i.i.d. standard Gaussian
 (use the same random seed at all nodes) for all $g \in [L]$
\STATE  Orthonormalize $\Urand$ using QR to get $\Urandorth$ for all $g \in [L]$
\STATE $\Upm^\sg \leftarrow \Urandorth$
\FOR{$\tau=1$ {\bfseries to} $T_{pm}$}
\STATE $\tilde\Upm_\inp^\sg \leftarrow ( \X_0)_\ssg  (\X_0)_\ssg{}^\top \Upm^\sg$, for all $g \in [L]$
\STATE ${\tilde\Upm}^\sg \leftarrow  \mathrm{AvgCons}_g( \tilde\Upm_\inp^\sgp, g' \in [L], \calG,  T_\con)$
\STATE QR on ${\tilde\Upm}^\sone \qreq \Upm^\sone \Rpm^\sone $ to get $\Upm^\sone$
\STATE $\Upm_\inp^\sone = \Upm^\sone, \Upm_\inp^\sg = \bm{0}, g \neq 1$
\STATE $\Upm^\sg \leftarrow \mathrm{AvgCons}_g( \Upm_\inp^\sgp, g' \in [L], \calG,  T_\con)$ for $g \neq 1$
\ENDFOR
\STATE Output $\U_0^\sg \leftarrow \Upm^\sg$
\end{algorithmic}
\end{algorithm}
%
\subsection{Minimization over $B$}
Starting with the spectral initialization of $\U$, at each iteration, our proposed Dec-MTRL algorithm alternates between the following two steps: (i)~Minimization over $\B$ and (ii)~Projected-GD (ProjGD) for $\U$. We present the detail of the minimization step below.
For each new estimate of $\U$, we solve for $\B$ by minimizing $f(\U, \B)$ over it while keeping $\U$ fixed at its current value. This is a decoupled column-wise least squares (LS) step.
Since the minimization step decouples for each column $\b_t$, this step is done locally at the node that has the corresponding $\y_t$. Thus, at node $g$, we compute $\b_t = (\A_t \U^\sg)^\dagger \y_t$ for all $t \in \S_g$. Each node thus updates a different $r \times |\S_g|$ sub-matrix $\B_\ssg = [\b_t, \ t \in \S_g]$. 
In this step, the most expensive part is the computation of the matrices $\A_t \U$ for all $t$, this takes time of order $ndr \cdot T$. Each LS computation only needs time of order $nr^2$.
%
\subsection{Decentralized Projected-GD for $U$}
 Recall the cost function 
in Eq.~\eqref{dec_cost_fn}.
The gradient computation for the GD step over $\U$ requires a sum over all $T$ terms.
In a decentralized setting, the estimates of $\U$ are different at the different nodes $g$, i.e., at node $g$, the cost function is $f_g(\U^\sg, \B_\ssg)$. Thus, node $g$ can only compute $\nabla f_g:= \nabla_\U f_g(\U^\sg,\B_\ssg)$.
$\nabla f_g$ is the input to AvgCons given in Algorithm \ref{alg:consensus},  at node $g$. It uses these inputs to compute an approximation to the sum, $\gradU := \sum_g \nabla f_g(\U^\sg, \B_\ssg)$. We denote the computed approximation (output of AvgCons) at node $g$ by $\hatgradU^\sg$. 
%
Each node then uses this in a GD step
\[
\Utilde_+^\sg \leftarrow \U^\sg - \eta \hatgradU^\sg,
\]
followed by QR decomposition to get 
$\Utilde_+^\sg \qreq \U_+^\sg \R_+^\sg.$
This is summarized in Algorithm \ref{altgdmin_dec}. We provide a guarantee for the GDmin steps in Theorem \ref{gd_thm}.
{\em In the rest of this paper, for simplicity, we remove the subscript $\U$ from the gradient notation. All gradients are w.r.t. the first argument of $f(.,.)$.}

\begin{algorithm}[t]
\caption{The complete Dec-MTRL algorithm. 
It calls AvgCons summarized in Algorithm \ref{avgcons_algo}.
}
\label{altgdmin_dec}
\begin{algorithmic}[1]
\STATE \textit {\bf Input:} $\A_t, \y_t$, $t \in \S_g$, $g \in [L]$ 
graph $\calG$

\STATE \textit {\bf Output:} $\U^\sg$, $\B_\ssg$ and $\X_\ssg  = \U^\sg \B_\ssg$.

\STATE \textit {\bf Parameters:} $\eta$, $T_\con$,  $T_\gd$

\STATE {\bfseries Sample-split:} Partition $\A_t, \y_t$ into  $2T_\gd+2$ equal-sized disjoint sets: $ \A_t^{(\ell)}, \y_t^{(\ell)}, \ell=00,0,1,2,\dots 2T_\gd$

\STATE {\bfseries Initialization:}
Run Algorithm~\ref{altgdmin_dec_init} to get $\U_0^\sg$.

\STATE {\bfseries AltGDmin iterations:}
%
\FOR{$k=1$ {\bfseries to} $T_\gd$} 
\STATE $\U^\sg \leftarrow \U_{k-1}^\sg$
\STATE Let $\y_t = \y_t^{(k)}, \A_t = \A_t^{(k)}$ 

\STATE $\b_t  \leftarrow  (\A_t \U^\sg)^\dagger \y_t$ $\forall t \in \S_g$ 

\STATE $\x_t    \leftarrow \U^\sg  \b_t$  $\forall t \in \S_g$

\STATE Let $\y_t = \y_t^{(T_\gd+k)}, \A_t =  \A_t^{(T_\gd+k)}$.

\STATE $\nabla f_g \leftarrow \sum_{t \in \S_g} \A_t^\top (\A_t \U^\sg \b_t - \y_t) \b_t^\top$ 

\STATE $\hatgradU^\sg \leftarrow  \mathrm{AvgCons}_g( \nabla f_{g'}, g' \in [L], \calG, T_\con)$

\STATE  $\Utilde^\sg_+ \leftarrow \U^\sg  - \eta \hatgradU^\sg$   \label{alg:GDstep}

\STATE  QR decompose $\Utilde^\sg_+ \qreq \U^\sg_+ {\R}^\sg_+$ to get $\U^\sg_+$  
\STATE  $\U^\sg_t \leftarrow \U^{\sg}_+$
      \ENDFOR
\STATE Output  $\U^\sg$, $\B_\ssg$ and $\X_\ssg  = \U^\sg \B_\ssg$
\end{algorithmic}
\end{algorithm}
\section{Main Results and Discussion} \label{results}
\subsection{Main Result}
{\cblue Our main result, Theorem~\ref{final_res}, stated next, says the following.
Consider Algorithm \ref{altgdmin_dec} with parameters (step size, $\eta$, number of consensus iterations at each GDmin or PM iteration, $T_\con$, total number of PM iterations for initializing Dec-MTRL, $T_{pm}$, and total number of GDmin iterations, $T_\gd$) set as specified in it.
Assume that the graph of the underlying communication network is connected, the right singular vectors of $\Xstar$ satisfy the incoherence assumption (Assumption \ref{right_incoh2}), and the $\A_t$'s are i.i.d. $n\times d$ random Gaussian matrices. If the number of samples per column, $n$, satisfies the specified lower bounds,  then, with high probability (whp), the subspace estimation error after $T_\gd$ GDmin iterations, $\SE_2(\U_T^\sg, \Ustar)$, is at most $\epsfin$. The same is true for the normalized error in estimating $\xstar_t$'s.}

Below we present our main result. 
\begin{theorem}\label{final_res}
Assume that Assumptions \ref{right_incoh2}, \ref{iidAk}, \ref{connectedcalG} hold. 
Consider Alg.~\ref{altgdmin_dec} initialized using Alg.~\ref{altgdmin_dec_init}.
Let $\eta = 0.4 /n \sigmax^2$, $T_{pm} = C \kappa^2 (\log d + \log \kappa)$, $T_\gd = C \kappa^2 \log(1/\epsfin)$,   $T_\con =C \Tconexp (T_\gd +  \log( 1/\epsfin) + \log d + \log L +\log \kappa)$.
Assume 
\[
 nT \ge C  \kappa^6 \mu^2 (d+T) r \left( \kappa^2 r + \log({1}/{\epsfin}) \right).
\]
Then, w.p. at least $1 - 1/d$,
\[
\SE_2(\U_{T_gd}^\sg, \Ustar) \le \epsfin, \text{ and }  \|\x_t - \xstar_t\| \le \epsfin \|\xstar_t\|,
\]
for all $t \in \S_g,~ g \in [L].$
\end{theorem}


The above result is what is often referred to as a {\em non-asymptotic} or {\em constructive} convergence guarantee: the specified values of  $T_\gd, T_\con$ enable us to also provide a bound on the  time and the communication complexity of Algorithm \ref{altgdmin_dec}. We obtain these below in Section~\ref{complexities}.
%
Consider the sample complexity (the required lower bound on $nT$). To understand it, suppose that $d \approx T$. Ignoring log factors, and treating $\kappa,\mu$ as numerical constants, our result implies that roughly order  $r^2$ samples suffice per task. If the low rank (LR) assumption and our algorithm were not used, we would have to invert each $\A_t$ to recover each $\xstar_t$ from $\y_t$;  this would require  $n \ge d$  samples per task instead of just $r^2$. Under the LR assumption ($r \ll d$, e.g., $r=\log d$), our required sample complexity is much lower than what inverting $\A_t$ would require.

%
\subsection{Time and Communication Complexity} \label{complexities}


The time complexity is $(T_{pm} \cdot T_\con) \cdot \varpi_\init + (T_\gd \cdot T_\con ) \cdot \varpi_\gd$ where $\varpi_\init$ and $\varpi_\gd$ are the time needed by one consensus iteration of the initialization or the GDmin steps.
In each GDmin consensus iteration, we need to (i) compute $\A_t \U$ for all $t \in \S_g$, $g \in [L]$, (ii) solve the LS problem for updating $\b_t$ for all $t \in \S_g$, $g \in [L]$, and (iii) compute the gradient w.r.t. $\U$ of $f_t(\U, \B)$, i.e. compute $\A_t^\top (\A_t \U \b_t - \y_t) \b_t^\top$ for all $t \in \S_g$, $g \in [L]$, and (iv) implement the GD step followed by a QR decomposition at all the nodes. The QR of an $d \times r$ matrix takes time of order $dr^2$ and this needs to be done at all the $L$ nodes. Recall that $\sum_g |\S_g|=T$. Thus, order-wise, the time taken per iteration is $\varpi_\gd = \max(T \cdot nd r , T \cdot nr^2 , T \cdot ndr, L \cdot dr^2) = nT dr $ since $L \le T$.
In the initialization step, the computation cost for computing the threshold $\alpha$ is negligible since it is a scalar. One PM iteration needs time of order $\varpi_\init = nTdr$.
Thus, the time complexity is $T_\con \cdot (T_\gd + T_{pm}) nTdr$ $= C \kappa^4   \Tconexp  \max\{\log^2(1/\epsfin), \log^2 d, \log^2 \kappa, \log^2 L\}nTdr$.

In all consensus iterations, the nodes are exchanging approximations to the gradient of the cost function w.r.t. $\U$. This is a matrix of size $d \times r$. In one such iteration, each node receives $nr$ scalars from its neighbors. Thus the communication complexity per iteration per node is $dr \cdot (\max_g \dg)$ where $(\max_g \dg)$ is the maximum degree of any node. The same is true for the initialization too.
Thus the cost per iteration for all the nodes is $dr \cdot (\max_g \dg) \cdot L$ and so the overall communication cost is $T_\con \cdot (T_\gd+ T_{pm} ) \cdot dr \cdot (\max_g \dg) \cdot L$.
%
\begin{corollary}
The time complexity of Algorithm \ref{altgdmin_dec} is $C \kappa^4   \Tconexp  \max\{\log^2(1/\epsfin), \log^2 d, \log^2 \kappa, \log^2 L\}$ $nTdr.$ Its communication complexity is
$C \kappa^4   \Tconexp$ $\cdot \max \{\log^2(1/\epsfin), \log^2 d, \log^2 \kappa, \log^2 L\}) \cdot dr (\max_g \dg).$
\end{corollary}
\begin{brem}
Assuming that the number of nodes $L$ is small (is a numerical constant, i.e., it does not grow with $d,T,r$), our time complexity is nearly linear in the problem size $nT dr$, thus making the algorithm fast. The communication complexity per node per iteration is order $dr$ which is the size of $\U$ that needs to be shared. Thus, this cannot be improved any further.
\end{brem}
%
\subsection{Proving Theorem~\ref{final_res}}
We can prove the following results for the GD and minimization (GDmin) steps, and the initialization respectively. Combining these proves  Theorem~ \ref{final_res}.
The GDmin result says the following. 
If the number of samples per task, $n$, satisfies the specified lower bounds,  if the initialization error is small enough, and if the initial estimates of $\Ustar$ at the different nodes are also close entry-wise (and hence in Frobenius norm), then, whp,
the subspace estimation error $\SE_2(\U_\t^\sg, \Ustar)$ decays exponentially. The same is true for the normalized error in estimating $\xstar_t$'s.
The initialization result helps prove that the initial estimates do indeed satisfy the required bounds. Combining the two results then proves Theorem~ \ref{final_res}.
%
%
\begin{theorem}[GDmin iterations] \label{gd_thm}
Assume that Assumptions \ref{right_incoh2}, \ref{iidAk}, \ref{connectedcalG} hold. 
Consider Algorithm \ref{altgdmin_dec}.
Pick a final desired error $\epsfin < 1$. Let $\eta = c_\eta /n \sigmax^2$ with $c_\eta=0.4$, $T_\gd = C \kappa^2 \log(1/\epsfin)$, and $T_\con =  C \Tconexp (\log L +  T_\gd +  \log( 1/\epsfin))$.
Assume  that, at each iteration $k$,
\[
nT \gtrsim \kappa^4 \mu^2 dr, \ n \gtrsim \max(\log d, \log T, r).
\]
If $\SE_2(\U_0^\sg, \Ustar) \le \delta_0= \frac{c }{\sqrt{r} \kappa^2}$, and $\max_{g' \neq g} \|\U_0^\sg - \U_0^\sgp\|_F \le \rho_0 =  c^{T_\gd}\epsfin/ \kappa^2$,  then, w.p. at least $1-T  d^{-10}$, at any iteration $k \le T_\gd$,
\[
\SE_2(\U_\t^\sg, \Ustar) \le \delta_\t:=(1 - \frac{c_\eta}{\kappa^2})^\t  \delta_0,  \ \|(\x_t)_\t - \xstar_t\|\le \delta_\t \|\xstar_t\|, 
\]
for all $t \in S_g, g \in [L].$
\end{theorem}


\begin{theorem}[Initialization]\label{init_step_PM_claim}
Assume that Assumptions \ref{right_incoh2}, \ref{iidAk}, \ref{connectedcalG} hold.
Pick a $\delta_0 < 1$ and a $\rho_0 < 1$.
Consider Algorithm \ref{altgdmin_dec_init}.
If $T_\con = C \Tconexp ( \log L + \log d + \log \kappa + \log (1/\delta_0) + \log(1/\rho_0) )$,
$T_{pm} = C \kappa^2 \log(\frac{d}{\delta_0})$, and
$nT \gtrsim \kappa^4 \mu^2 (d+T) \frac{r}{\delta_0^2},$
then, w.p. at least $1 - 1/d$,
$\SE_2(\U^\sg_0, \Ustar) \le \delta_0,  \text{ and } \max_{g' \neq g} \|\U_0^\sg - \U_0^\sgp\|_F \le \rho_0.$
\end{theorem}
\subsubsection{Proof of Theorem~ \ref{final_res}}
To combine our guarantees for GDmin iterations and initialization,  we set
$\delta_0 = c / \sqrt{r} \kappa^2,  \text{ and } \rho_0 =  c^{T_\gd} \epsfin / \kappa^2.$
This means, for initialization, we need $nT \gtrsim \kappa^8 \mu^2 (d+T) r^2:= n_{\init} T$, $T_\con = C \Tconexp (\log L  + T_\gd +  \log( 1/\epsfin) + \log d +  \log \kappa )$, and $T_{pm}=C \kappa^2 \log(\frac{d}{\delta_0})$  $=C \kappa^2 \log(d \sqrt{r} \kappa^2) = C \kappa^2 (\log d + \log \kappa)$ since $r \le d$.
%
GDmin iterations need $nT\gtrsim \kappa^4 \mu^2 dr:= n_{\gd} q$. 
Thus, the total number of samples $nT$ needed is $ 
nT \gtrsim   n_\init T + T_\gd n_{\gd} T. 
$
By plugging in $T_\gd =  C \kappa^2 \log(1/\epsfin)$, we get the sample complexity.
\qed

\section{Proving Theorem \ref{gd_thm}:GDmin Iterations}\label{sec:proof}
{\cblue 
The proof of Theorem \ref{gd_thm} is derived from the combination of Theorem~\ref{gd_step_claim} and Lemma~\ref{conserr_bnd_lemma}. We prove these results in the next subsections.}
\subsection{Proving Theorem~\ref{gd_step_claim}: GD step and Consensus for $U$}
{\cblue To prove Theorem~\ref{gd_step_claim}, we first prove Theorem~\ref{min_step_claim} for the minimization step and then use this result to analyze the GD step in Theorem~\ref{gd_step_claim}.}
\begin{align*}
\mbox{Define}~\B & = [\B_\ssone, \B_\sstwo, \dots, \B_\ssL], \text{ with } \B_\ssg: = [\b_t, \ t \in \S_g] \\
\X & = [\X_\ssone, \X_\sstwo, \dots, \X_\ssL], \text{ with } \X_\ssg: = \U^\sg \B_\ssg
\end{align*}
and $\x_t = \U^\sg \b_t$ for all $t \in \S_g$.
Define
$
\g_t: = \U^\sg{}^\top \xstar_t,  \ \forall \ t \in \S_g
$
and
\begin{align*}
\G & = [\G_\ssone, \G_\sstwo, \dots, \G_\ssL], \G_\ssg: = [\g_t, \ t \in \S_g] = \U^\sg{}^\top \X_\ssg^\star\\
\DD & = [\DD_\ssone, \DD_\sstwo, \dots, \DD_\ssL],  \DD_\ssg: = \U^\sone{}^\top \X_\ssg^\star.
\end{align*}
Thus,  $\DD = \U^\sone{}^\top \Xstar$.
Notice the difference between $\G$ and $\DD$.  Both will be used in our proof.
%
In this proof, we use the following lemma from \cite{lrpr_gdmin}.

\begin{lemma}[\cite{lrpr_gdmin}, Lemma 3.3, first part]
\label{bhat_lemma_2}
Let 
 $\g_t: = \U^\sg{}^\top \xstar_t.$
Then, w.p. $\ge 1 - T \exp(r - c n)$,
\begin{align}
\|\g_t - \bhat_t \|   & \leq 0.4 \|\left(\I_d-\U^\sg\U^\sg{}^\top \right)\Ustar\bstar_t\|, \mbox{~for~} t \in \S_g. \nn
\end{align}
\end{lemma}
%

%

\begin{theorem}[Min step for $\B$] \label{min_step_claim}
%
Assume that, for all $g$, $\SE_2(\U^\sg,\Ustar) \le \delta_k$, for some $\delta_k \geqslant 0$. 
If $n \gtrsim \max(\log T, \log d, r)$, then, w.p $\geqslant 1-\exp(\log T+ r -cn)$
\ben
\item $\|\b_t - \g_t\| \le 0.4 \delta_\t \|\bstar_t\|$

\item $\|\b_t\| \le \|\g_t\|+0.4 \cdot 0.02 \|\bstar_t\| \le 1.1 \|\bstar_t\|$

\item $\|\x_t - \xstar_t\| \le 1.4  \delta_\t \|\bstar_t\|$

\item $\|\B - \G\|_F \le 0.4 \delta_\t \|\Bstar\|_F \le  0.4 \sqrt{r} \delta_\t \sigmax $

\item $\|\X - \Xstar\|_F \le 1.4 \sqrt{r} \delta_\t \sigmax $
\een
If $\delta_\t \le 0.02 / \sqrt{r} \kappa^2$ and if  $\max_{g' \neq g} \|\U^\sg - \U^\sgp\|_F \le \tb_\t$ with $\tb_\t \le 0.1 \delta_\t$, then the above implies that
\ben
\item[a)] $\sigmamin(\B) \ge  0.9 \sigmin$ and \item[b)] $\sigma_{\max}(\B) \le  1.1 \sigmax$.
\een
\end{theorem}
\begin{proof}
By Lemma \ref{bhat_lemma_2}, if $n \gtrsim \max(\log T, \log d, r)$, then,  w.p $\geqslant 1-\exp(\log T+ r -cn)$
 \[
\|\b_t{} - \g_t{}\| \le 0.4 \|(\I_d - \U^\sg \U^\sg{}^\top) \Ustar \bstar_t\|.
\]
Under the assumption that $\SE_2(\U^\sg,\Ustar) \le \delta_\t$ with $\delta_\t < 0.02$, this directly implies the following 
\ben
\item $\|\b_t - \g_t\| \le 0.4 \delta_\t \|\bstar_t\|$

\item $\|\b_t\| \le \|\g_t\|+0.4 \cdot 0.02 \|\bstar_t\| \le 1.1 \|\bstar_t\|$

\item $\|\x_t - \xstar_t\| \le 1.4  \delta_\t \|\bstar_t\|$ \\ (follows by adding/subtracting $\U^\sg \g_t$ and using triangle inequality).
\een
Using above,
\[
\|\B - \G\|_F \le 0.4 \delta_\t \sqrt{\sum_t \|\bstar_t\|^2}  \le 0.4 \delta_\t \sqrt{r} \sigmax
\]
and similarly for $\|\X - \Xstar\|_F$. Thus,
\ben
\item[4.] $\|\B - \G\|_F \le 0.4 \delta_\t \|\Bstar\|_F \le  0.4 \delta_\t \sqrt{r} \sigmax $

\item[5.] $\|\X - \Xstar\|_F \le 1.4 \sqrt{r} \delta_\t \sigmax $

\een
To bound the singular values of $\B$, we show next that $\B$ is close to $\DD$ and bound the singular values of $\DD$. To do this, we use the fact that every $\U^\sg$ is close to $\U^\sone$ to show that $\G$ is close to $\DD$. Since we have already bounded $\B -\G$ this implies a bound on $\B - \DD$.

Recall from the theorem statement that $\tb_\t$ is the upper bound on $\max_{g' \neq g} \|\U^\sg - \U^\sgp\|_F$.
We have
\begin{align*}
\|\G - \DD\|_F^2 \le \max_g \|\U^\sg - \U^\sone\|_F^2 \sum_g \|\Xstar{}^\sg\|_F^2 \le \tb_\t^2 r \sigmax^2.
\end{align*}
Thus, using the above and the bound on $\B - \G$,
\[
\|\B - \DD\|_F \le (0.4 \delta_\t  + \tb_\t) \sqrt{r}\sigmax.
\]
Next, we have
\[
\sigmamin(\DD) \ge \sigmamin(\U^\sone{}^\top \Ustar )\sigmamin(\Bstar)=\sigmamin(\U^\sone{}^\top \Ustar ) \sigmin,
\]
\begin{align*}
\sigmamin^2(\U^\sone{}^\top \Ustar)
 & = \lambda_{\min}( \Ustar^\top \U^\sone \U^\sone{}^\top\Ustar )  \\
& = \lambda_{\min}(\I - \Ustar^\top \Ustar + \Ustar^\top \U^\sone \U^\sone{}^\top\Ustar) ) \\
& = \lambda_{\min}(\I - \Ustar^\top (\I -  \U^\sone \U^\sone{}^\top ) \Ustar )  \\
& = 1 - \lambda_{\max}(\Ustar^\top (\I -  \U^\sone \U^\sone{}^\top ) \Ustar )  \\
& = 1 - \lambda_{\max}(\Ustar^\top (\I -  \U^\sone \U^\sone{}^\top )^2 \Ustar )  \\
& = 1 - ||(\I -  \U^\sone \U^\sone{}^\top ) \Ustar||^2 
\ge 1 - \delta_\t^2.
\end{align*}
We used $\Ustar^\top \Ustar=\I$ and $(\I -  \U^\sone \U^\sone{}^\top ) = (\I -  \U^\sone \U^\sone{}^\top )^2$.
Also,
\[
\sigmamin(\B) \ge \sigmamin(\DD) - \|\B - \DD\| \ge \sigmamin(\DD) - \|\B - \DD\|_F.
\]
%
Combining the above three bounds, if $\delta_\t < 0.02 / \sqrt{r} \kappa^2$, then
\[
\sigmamin(\B) \ge \sqrt{1- \delta_\t^2}\sigmin -  (0.4 \delta_\t  + \tb_\t) \sqrt{r}\sigmax > 0.9\sigmin,
\]
if $\tb_\t \le 0.1 \delta_\t$ and $\delta_\t \le 0.1/\sqrt{r}\kappa$.
Finally, since $\|\DD\| \le \|\Bstar\| = \sigmax$, under the above bounds on $\delta_\t, \tb_\t$,
\[
\sigma_{\max}(\B) \le  \sigmax + (0.4 \delta_\t  + \tb_\t) \sqrt{r}\sigmax \le 1.1 \sigmax.
\]
We have thus proved Theorem \ref{min_step_claim}.
\end{proof}
To prove the GD step, define
\begin{align*}
\nabla f_g & : = \nabla f_g(\U^\sg, \B_\ssg)  : = \sum_{t \in \S_g} \nabla_\U f_t(\U^\sg, \b_t ) \\
\hatgradU^\sg & := \text{AvgCons}_g( \nabla f_{g'}(\U^\sgp, \B_\ssgp), g' \in [L] ) \\
\gradU & := \sum_{g'} \nabla f_{g'}(\U^\sgp, \B_\ssgp)
\end{align*}
$\nabla f_g$ is the input to the AvgCons algorithm at node $g$ and $\hatgradU^\sg$ is an estimate of $\gradU$ obtained by running $T_\con$ iterations of the consensus algorithm AvgCons at node $g$.

Define the error terms
\begin{align}
\conserr^\sg &  :=  \gradU -\hatgradU^\sg  \nn \\
\err & :=  \E[\gradU]- \gradU \nn \\
\Uerr^\sg & :=   \sum_{g'} ( \U^\sg -\U^\sgp) \B_\ssgp \B_\ssgp{}^\top
\label{def_errterms}
\end{align}
and define
$\P:= \I - \Ustar \Ustar^\top.$
Now, we present the GD claim in Theorem~\ref{gd_step_claim}.
\begin{theorem}[GD step and consensus for $\U$] \label{gd_step_claim}
Assume that $\eta = c_\eta / n \sigmax^2 $ with $c_\eta \le 0.5$.
 Assume that $\SE_2(\U^\sg,\Ustar) \le \delta_\t$, $\|\U^\sg - \U^\sgp\| \le \tb_\t$,  $\|\conserr^\sg\|_F \le n \eps_\con  \sigmin^2$.

If $\delta_\t \le 0.02 / \sqrt{r} \kappa^2$,  $\tb_\t < 0.1 \delta_\t/\kappa^2$, $\eps_\con < 0.1 \delta_\t / \kappa^2$,
and if $nT \ge C \kappa^4 \mu^2 dr $ and $n \gtrsim \max(\log d, \log T, r)$  at iteration $k$,  then, whp,     
\[ 
\SE_2(\U^\sg_+, \Ustar) \le \delta_{\t+1}:=   \delta_\t (1 - 0.6 \frac{c_\eta}{\kappa^2})\mbox{~and}
\]
\[
\max_{g'} \|\U^\sg_+ - \U^\sgp_+\|  \le \tb_{\t+1}:= 1.7 (\tb_\t + (2/\kappa^2)  c_\eta \eps_\con)
\]
\end{theorem}
\begin{proof}
Combining Lemma~\ref{gd_lemma} and Lemma~\ref{bt_lemma}, the proof of Theorem~\ref{gd_step_claim} follows. The proofs of the lemmas are presented in Section~\ref{sec:supp_lem}.
\end{proof}
%
\subsection{Proving Lemma \ref{conserr_bnd_lemma}: Bounding $\conserr$}
We first present the following supporting lemma.
\begin{lemma}\label{avgcons_prop}
If
\[
\scalemath{0.9}{T_\con \ge \frac{3}{2} \Tconexp \log (L \sqrt{L} /\eps_\con ) = C \dfrac{\log (L /\eps_\con )}{\log(1/\gamma(\W))},}
\]
then
\[
\|\conserr^\sg\|_F \le \eps_\con \max_g \|\inperr^\sg\|_F
\]
Here $\inperr^\sg$ is the error between input of node $g$ to the consensus algorithm and the desired sum to be computed and $\conserr^\sg$ is the error between the final output at node $g$ and $T_\con$ iterations and the desired sum.
\end{lemma}
\begin{proof}
    Recall that $\conserr^\sg:= \hatgradU^\sg - \gradU $ and $\hatgradU^\sg$ is the output of our AvgCons algorithm applied to each entry of $\nabla f_g$. Define
\[
\inperr^\sg  := \nabla f_g - \gradU   =  \sum_{g' \neq g}\nabla f_{g'}  \nn \\
\]
with $\nabla f_g$ being the partial gradient sum at node $g$. This is the input to the AvgCons algorithm.

We bound $\conserr$ in terms of $\inperr$ by using Proposition \ref{avgcons_prop_scalar} for scalar consensus.
%
We apply this result to each of the $dr$ scalar entries of $\nabla f_g$. Thus, for the $(j,j')$-th entry, $z_{true} = (\gradU)_{j,j'}$, $z_{\inp}^\sg = (\nabla f_g)_{j,j'}$,
and $z_{out}^\sg = (\hatgradU^\sg)_{j,j'}$. Assume the graph of the network is connected.
%
%
Using this result, if  $T_\con \ge \Tconexp \log L/\epsconscalar $
then
\[
\max_{g} |\conserr_{j,j'}^\sg| \le \epsconscalar \max_g |\inperr_{j,j'}^\sg|
\]
Using $\sum_{j,j'} \max_g |\inperr_{j,j'}^\sg|^2 \le \sum_{j,j'} \sum_g |\inperr_{j,j'}^\sg|^2 = \sum_g \|\inperr^\sg\|_F^2 \le L \max_g \|\inperr^\sg\|_F^2 $,
\[
\max_g \|\conserr^\sg\|_F \le \epsconscalar \cdot \sqrt{L} \max_g \|\inperr^\sg\|_F
\]
By setting $\epsconscalar = \eps_\con / \sqrt{L}$, we complete the proof.
\end{proof}

%

The next lemma bounds consensus error. We prove it by obtaining a bound on $\inperr^\sg$. To do this, we first bound $\|\E[\inperr^\sg]\|$ and then bound $\|\inperr^\sg - \E[\inperr^\sg]\|$ as done in the proof of Lemma \ref{grad_bnd_new2}: use sub-exponential Bernstein inequality followed by a standard epsilon-net argument.
To bound $\|\E[\inperr^\sg]\|_F$, we also need the following claim.
\begin{claim} \label{norm_cols_removed}
If $\M'$ is $\M$ with some columns removed, then $\|\M'\| \le \|\M\|$.
\end{claim}

\begin{lemma} \label{conserr_bnd_lemma}
Assume that,  for all $g$, $\SE_2(\U^\sg,\Ustar) \le \delta_\t$ and $\max_{g \neq g'} \|\U^\sg - \U^\sgp\| \le \tb_\t$. Also, $nT \ge  C \kappa^4 \mu^2 d r$.
If $T_\con =  \Tconexp \log (L /\eps_\con )$, then, whp, $\|\conserr^\sg\|_F \le \eps_\con n \sigmin^2$.
\end{lemma}
\begin{proof} For a matrix $\Z$, let $\Z_{\setminus g}$ denotes a submatrix of $\Z$ obtained after eliminating the columns $t \in \S_g$.
Using Claim \ref{norm_cols_removed} and bounds on $\|\B\|$ and  $\|\Xstar - \X\|_F$ from  Theorem \ref{min_step_claim}, if $\delta_\t < \frac{c}{\sqrt{r} \kappa^2}$
\begin{align}
\| \E[\inperr^\sg]\|_F   & =  \|\sum_{g'\neq g}\sum_{t \in \S_{g'}}n (\x_t-{\xstar_t})\b_t{}^\top \|_F  \nn\\
&= n\|(\X_{\setminus g} -\Xstar_{\setminus g} )\B_{\setminus g}{}^\top \|_F  \nn\\
&\le n\|(\X_{\setminus g} -\Xstar_{\setminus g}\|_F \cdot \| \B_{\setminus g}{}^\top \|  \nn
\end{align}
\begin{align}
&\le n\|\X - \Xstar\|_F \cdot \| \B\|    \le 1.1 n\delta_\t \sqrt{r}  \sigmax^2. \nn
\end{align}
w.p. $1- \exp(\log T+r-cn)$. The last step used Theorem \ref{min_step_claim}.

Proceeding exactly as in the proof of Lemma \ref{grad_bnd_new2}, we can show that, if $\delta_\t \le c/\sqrt{r} \kappa^2$,  $\tb_\t \le 0.1 \delta_\t/\kappa^2$, and $n$ satisfies the bounds stated in that lemma, then,
\[
	 \max_g \|\inperr^\sg - \E[\inperr^\sg] \| \le  1.1 \eps_1 \delta_\t n \sigmin^2
\]
w.p. at least  $1-   \exp ( C(d+r) -c\eps_1^2nT/ \kappa^4 \mu^2 r )- \exp(\log T+r-cn)$.

Thus, with the above probability, $\|\inperr^\sg - \E[\inperr^\sg] \|_F \le  \sqrt{r} \cdot 1.1 \eps_1 \delta_\t  n \sigmin^2$.
Setting $\eps_1 =0.1$, we conclude that
w.p. at least  $1-  L \exp ( C(d+r) -c \frac{nT}{ \kappa^4 \mu^2 r} )- \exp(\log T+r-cn)$,
\[
\max_g  \|\inperr^\sg\|_F \le (1.1 + 0.11) n (\delta_\t \sqrt{r}) \sigmax^2 \le 0.1 n \sigmin^2
\]
if $\delta_\t$ and $\tb_\t$ are bounded as stated above.
Combining this with Lemma \ref{avgcons_prop}, we have proved the lemma. 
 \end{proof}

\subsection{Proof of Theorem \ref{gd_thm}}
\begin{proof}[Proof of Theorem \ref{gd_thm}]
 Simplifying the recursion of Thm.\ref{gd_step_claim}, 
\[
\tb_{\t+1} = 1.7^{\t+1} \tb_0 + 3.3 \t 1.7^{\t} \cdot c_\eta  \eps_\con/\kappa^2, \text{ and }
\]
\[
\delta_{\t+1} = (1  - 0.6 \frac{c_\eta}{\kappa^2})^{\t+1} \delta_0
\]

Since $\delta_\t$ decreases with $\t$ and $\tb_\t$ increases with $\t$, thus, the bounds of Theorem \ref{gd_step_claim} hold if $\delta_0 \le c / \sqrt{r} \kappa^2$, $\tb_{T_\gd} \le 0.1 \delta_{T_\gd}/\kappa^2 $,  and $\eps_\con < 0.1 \delta_{T_\gd} / \kappa^2 $.
Using the expressions above, in order to ensure $\delta_{T_\gd} = \epsfin$, we need
\ben

\item $T_\gd \ge C  \frac{\kappa^2}{c_\eta} \log(1/\epsfin)$,

\item  $\delta_0 \le c / \sqrt{r} \kappa^2$,
\item $\eps_\con < 0.1  \epsfin/\kappa^2$,
\item
$
\tb_{T_\gd} = 1.7^{T_\gd} 1.7 \tb_0 + 3.3 \cdot  T_\gd \cdot 1.7^{T_\gd} \cdot c_\eta  \eps_\con/\kappa^2 \le 0.1 \epsfin / \kappa^2.
$
This holds if
\ben
\item[a)]  $\tb_0 \le c^{T_\gd} \epsfin / \kappa^2$, with $c=1/1.8$, and 
\item[b)] $c_\eta \eps_\con \le c^{T_\gd} \epsfin / (T_\gd).$ 
\een
\een
Setting $c_\eta = 0.4$ (or any constant $\le 0.5$), we thus need
(1) $T_\gd = C \kappa^2 \log(1/\epsfin)$ and
(2) $\eps_\con  = c^{T_\gd} \epsfin /  (T_\gd)$.

By Lemma \ref{conserr_bnd_lemma}, if $T_\con =  \Tconexp \log (L /\eps_\con )$, then $\|\conserr^\sg\|_F \le \eps_\con  n \sigmin^2$.
Thus, to get $\eps_\con = c^{T_\gd} \epsfin /  (T_\gd)$, we need to set
\[
T_\con 
=  C \Tconexp (\log L +  T_\gd + \log T_\gd + \log( 1/\epsfin))
\]

By using this expression and the above and Theorem \ref{gd_step_claim}, we have proved Theorem \ref{gd_thm}.
\end{proof}
\begin{brem}
We get the $t 1.7^\t \eps_\con $ type factor in the expression for $\tb_{\t+1}$ given above because of the projected GD step. If there was no projection (simple GD), we would get a factor of only $\t \eps_\con $. This would eliminate the need for $T_\con$ to depend on $T_\gd$, it would instead only depend on $\log T_\gd$.
%
\end{brem}
\subsection{Supporting Lemmas for proving Theorem~\ref{gd_step_claim}}\label{sec:supp_lem}
Assume that, for all $g$.
\bi
\item $\|\conserr^\sg\|_F \le n \eps_\con \sigmin^2$ for all $g$

\item $\SE_2(\U^\sg, \Ustar) \le \delta_\t$
\item $\|\U^\sg - \U^\sgp\|_F \le \tb_\t$.

\ei

The result in this section uses Theorem \ref{min_step_claim} from the previous section. We first present and prove the supporting lemmas.
\begin{lemma}[Improved version of gradient deviation lemma from \cite{lrpr_gdmin}]
\label{grad_bnd_new2}
Assume that,  for all $g$, $\SE_2(\U^\sg,\Ustar) \le \delta_\t$ and $\max_{g \neq g'} \|\U^\sg - \U^\sgp\| \le \tb_\t$.
The following hold:
\ben
\item $\E[\gradU] = m (\X - \Xstar) \B^\top$
\item $\|\E[\gradU]\| \le 1.1 n\delta_\t \sqrt{r}  \sigmax^2$
\item
If $\delta_\t < \frac{c}{\sqrt{r} \kappa}$, and $\tb_\t < 0.1 \delta_\t$
then,
\\
w.p. at least $1-  \exp (C (d+r) -c\frac{\eps_1^2 n T }{\kappa^4 \mu^2 r} )- \exp(\log T+r-cn)$
 \[
\|\gradU - \E[\gradU]\| \le  \eps_1 \delta_\t n \sigmin^2.
\]
\een
\end{lemma}
\begin{proof}
See  Appendix \ref{grad_bnd_new2_proof}  in \cite{TAC_arxiv} (arxiv version of the paper).
Proof follows similar to that in the centralized case. The key difference is that, in the $\gradU$ expression, different matrices $\U^\sg$ are used at the different nodes. 
The centralized case bound itself is a strengthening of what was proved in \cite{lrpr_gdmin} (Lemma 3.5, part 1). 
%
Differences from \cite{lrpr_gdmin} (Lemma 3.5, part 1):  (i) We use $\SE_2$ instead of $\SE_F$ as our subspace distance measure, and hence we use Theorem \ref{min_step_claim} instead of the corresponding result used in  \cite{lrpr_gdmin}. (ii) Because of this, there is one key change in how we bound the probability while applying the sub-exponential Bernstein inequality. These two changes in fact allow us to get a bound with a better sample complexity than the one proved in \cite{lrpr_gdmin}.
 \end{proof}
\begin{lemma}\label{gd_lemma}
If $\delta_\t < c / \sqrt{r} \kappa^2$,  $\tb_\t < 0.1 \delta_\t/\kappa^2$,  $\eps_\con < 0.1 \delta_\t/\kappa^2$,  if $\eta = c_\eta /n \sigmax^2$ with $c_\eta < 0.5$, and if $nT \gtrsim  \kappa^4 \mu^2 dr$ and $n \gtrsim \max(\log d, \log T, n)$,
then, whp, 
\[
\SE_2(\U^\sg_+, \Ustar) \le \delta_{\t+1}:= (1 - 0.6 \frac{c_\eta}{\kappa^2}) \delta_\t
\]
\end{lemma}
\begin{proof}
Recall that the GD step  of the algorithm is
\begin{align}
\Utilde^\sg_+ &= \U^\sg - \eta \hatgradU^\sg  \nn \\
\Utilde^\sg_+ &\qreq  \U^\sg_+ \R^\sg_+.
\label{GDstep}
\end{align}
Since $ \U_+= \Utilde_+ (\R_+)^{-1}$ and since $\|(\R_+)^{-1}\| =1/ \sigmamin(\R_+) = 1/ \sigmamin(\Utilde_+)$, thus,
\begin{eqnarray}
\SE_2(\U_+^\sg, \Ustar) & = \| \P \U_+^\sg\| \le \dfrac{\|\P \Utilde_+^\sg\|}{\sigmamin(\Utilde_+^\sg)}.
\label{SDeq}
\end{eqnarray}

Consider the numerator.
Adding/subtracting $\gradU$ and then also adding/subtracting $\E[\gradU]$, we get
\[
\Utilde^\sg_+ = \U^\sg - \eta \E[\gradU] + \eta \err + \eta \conserr^\sg.
\]
Using the first claim of Lemma \ref{grad_bnd_new2}, and projecting orthogonal to $\Ustar$,
\begin{align*}
& \P \Utilde^\sg_+ \\
& = \P \U^\sg - \eta n \P \X \B^\top + \eta  \P (\err + \conserr^\sg)
\end{align*}
since $\P \Xstar = \bm{0}$.
Now
\[
 \X \B^\top =  \sum_{g'} \U^\sgp \B_\ssgp \B_\ssgp{}^\top.
 \]
Adding/subtracting $ \U^\sg \sum_{g'} \B_\ssgp \B_\ssgp{}^\top$ gives
\[
  \X \B^\top
 =  \U^\sg \sum_{g'} \B_\ssgp \B_\ssgp{}^\top -  \Uerr^\sg
\]
Using above and the fact that $\sum_{g'} \B_\ssgp \B_\ssgp{}^\top = \B \B^\top$,
\begin{align}
 \P \Utilde^\sg_+ = & \P \U^\sg (\I - \eta n \B \B^\top) +  \nn \\
&\hspace*{-5 mm}   \eta  \P(n\Uerr^\sg +   \err +  \conserr^\sg)
\label{PUtilde}
\end{align}

By our assumption in the beginning of this section (also the assumption made in the Theorem),
\begin{align*}
\|\conserr^\sg\| & \le \eps_\con n \sigmin^2, \\
 \|\Uerr^\sg\| & \le \tb_\t \|\B\|^2 \le  \tb_\t    \sigmax^2
\end{align*}
where the last inequality used Theorem \ref{min_step_claim}.

By Lemma \ref{grad_bnd_new2} with $\eps_1 = 0.1$, if $\delta_\t \le 0.1/\sqrt{r} \kappa$, $\tb_\t \le 0.1 \delta_\t$, and $nT \ge C \kappa^4 \mu^2 dr $ and $n \gtrsim \max(\log d, \log T, r)$, then,
\[
\|\err\| \le 0.1 n \delta_\t\sigmin^2.
\]

Consider the first term of Eq.~\eqref{PUtilde}. Using Theorem \ref{min_step_claim}, 
\[
\lambda_{\min}(\I -  \eta n \B \B^\top) = 1 - \eta n \|\B\|^2 \ge 1 - 1.2 \eta n \sigmax^2.
\]
Thus, if $\eta < 0.5/ n \sigmax^2 $, then the above matrix is p.s.d. This along with Theorem \ref{min_step_claim} then implies that
\[
\|\I -  \eta n \B \B^\top\| = \lambda_{\max}(\I -  \eta n \B \B^\top) \le 1 - 0.9 \eta n  \sigmin^2
\]

Thus, if $\eta = c_\eta /n \sigmax^2$ with $c_\eta < 0.5$, $ \delta_\t < 0.01 / \sqrt{r} \kappa^2$, $\tb_\t  \kappa^2 < 0.1 \delta_\t$,  $\eps_\con  < 0.1 \delta_\t$, and if $nT \gtrsim \kappa^4 \mu^2 dr$ and $n \gtrsim \max(\log d, \log T, r)$, then, $ \|\P \Utilde^\sg_+\|$
\begin{align*}
 & \le \|\P \U^\sg\| (1 - 0.9\frac{c_\eta}{\kappa^2}) + \eta n (\tb_\t  \kappa^2  +  0.1 \delta_\t  + \eps_\con )\sigmin^2 \\
& \le \delta_\t (1 - 0.9\frac{c_\eta}{\kappa^2}) + \frac{ c_\eta}{\kappa^2} (\tb_\t  \kappa^2  +  0.1 \delta_\t  + \eps_\con  ) \\
& \le \delta_\t (1 - 0.9\frac{ c_\eta}{\kappa^2}) + \frac{ c_\eta}{\kappa^2} 0.3 \delta_\t = \delta_\t (1 - 0.6\frac{c_\eta}{\kappa^2}).
\end{align*}
The above used  $\|\P \U^\sg\| = \delta_\t$.
%
Next, we obtain a lower bound on $\sigmamin(\tilde\U^\sg_+)$. If $ \delta_\t < 0.01 / \sqrt{r} \kappa^2$, $\tb_\t   < 0.1 \delta_\t$ (needed for using the bound on $\err$ from Lemma \ref{grad_bnd_new2}), $\eps_\con \kappa^2 < 0.1 \delta_\t$, and the lower bounds on $m$ from above hold, then
\begin{align}
 \sigmamin(\Utilde^\sg_+) \nn & \scalemath{0.93}{\ge  {\sigmamin(\U^\sg) - \eta (\|\E[\mathrm{gradU}]\| + \|\err \| + \|\conserr^\sg\|)}} \nn \\
& \ge {1 - c_\eta \delta_\t \sqrt{r} \kappa^2 (1.4  + \frac{0.1}{\kappa^2 \sqrt{r}} + \frac{\eps_\con }{\sqrt{r} \delta_\t }  )   } \nn  \\
& \ge {1 -  1.7 c_\eta \delta_\t \sqrt{r}}\label{eq:kappar}
\end{align}
The above used the bound on $\|\E[\gradU]\|$ from Lemma \ref{grad_bnd_new2} and  
$\kappa^2\sqrt{r} >1$ and $\eps_\con < 0.1 \delta_\t$.

Using the last two bounds above and Eq.~\eqref{SDeq},
\begin{align}
& \SE_2(\U^\sg_+, \Ustar)
\le \frac{\delta_\t (1 - 0.6\frac{c_\eta}{\kappa^2})}{1 -  1.7 c_\eta \delta_\t \sqrt{r}} \nn \\
& \le \delta_\t (1 - 0.6\frac{c_\eta}{\kappa^2}) (1 + 3.4 c_\eta \delta_\t \sqrt{r} ) \nn  \\
& \le \delta_\t (1 - \frac{c_\eta}{\kappa^2} (0.6 - 3.4  \delta_\t \sqrt{r} \kappa^2) )  \le \delta_\t (1 - 0.6\frac{c_\eta}{\kappa^2}). \nn
\end{align}
We used $(1-x)^{-1} < (1+2x)$ if $|x| < 1/2$, and $ \delta_\t < 0.01 / \sqrt{r} \kappa^2$ (last row).
Thus, we have proved Lemma~\ref{gd_lemma}.
\end{proof}
We use the following result \cite{stewart1977perturbation} [Theorem~3.1], \cite{sun1995perturbation} [Corollary 4.2].
\begin{prop}[Perturbed QR factorization, Corollary 4.2 of \cite{sun1995perturbation}]\label{pqr_res}
Let $\tilde\Z_1 \qreq \Z_1 \R_1$ and $\tilde\Z_2 \qreq \Z_2 \R_2$.
Then,
\[
\|\Z_2 - \Z_1\|_F  \le \sqrt{2} \frac{\|\tilde\Z_2 - \tilde\Z_1\|_F}{\sigmamin(\tilde\Z_1)}
\]
if the RHS is less than $\sqrt{2} \cdot 4/\sqrt{10}$.
\end{prop}

\begin{lemma} \label{bt_lemma}
Assume everything from Lemma \ref{gd_lemma}. Then,
\[
\|\U_+^\sg - \U_+^\sgp\| \le \tb_{\t+1}:=  1.7 (\tb_\t + (2 /\kappa^2)  c_\eta \eps_\con)
\]
\end{lemma}
\begin{proof}
Using Eq.~\eqref{GDstep},  adding/subtracting $\gradU$, and recalling that $\tb_\t$ is the upper bound on $\|\U^\sg- \U^\sgp\|_F$,
\begin{align*}
& \|\tilde\U^\sg - \tilde\U^\sgp\|_F \\
& \le \|\U^\sg- \U^\sgp\|_F + \eta \|\conserr^\sg - \conserr^\sgp\|_F \\
& \le \tb_\t + \eta \cdot 2 \cdot \| \conserr^\sg \|_F
\end{align*}
Using Proposition \ref{pqr_res} and Eq.~\eqref{eq:kappar}, under the assumptions made for it,
\begin{align*}
& \|\U_+^\sg - \U_+^\sgp\|_F
  \le 1.5\frac{\|\Utilde_+^\sg - \Utilde_+^\sgp\|_F}{\sigmamin(\Utilde_+^\sg)}  \\
& \le   1.5 \frac{(\tb_\t + \eta \cdot 2 \cdot \|\conserr^\sg\|_F)}{1 -  1.7 c_\eta \delta_\t \sqrt{r}}  \le 1.5 \frac{(\tb_\t + (2/\kappa^2)   c_\eta   \eps_\con)}{1 - 0.1 c_\eta}.
\end{align*}
Using $c_\eta \le 0.5$ in the denominator, this gives 
\[
\|\U_+^\sg - \U_+^\sgp\|_F   \le 1.1 \cdot 1.5  (\tb_\t + (2/ \kappa^2)  c_\eta  \eps_\con) :=\tb_{\t+1}
\]

We have thus proved the Lemma~\ref{bt_lemma}.
\end{proof}
%
\section{Proving Theorem \ref{init_step_PM_claim}: Initialization}
\label{init_PM_thm_proof}  

{\cblue The overall idea is as follows. We obtain bounds for node 1. We show that all other nodes' subspace basis estimates are within an $\eps_\con$ distance of node 1, i.e. $\|\Upm_\tau^\sg - \Upm_\tau^\sone\|_F \le \eps_\con$ at each PM iteration $\tau$, including the final one. This implies a similar bound on  $\SE(\Upm_\tau^\sg, \Upm_\tau^\sone)$. This follows since
$\SE(\Upm_1,\Upm_2) \le 2\|\Upm_1 - \Upm_2\|_F$.

Let $\Upm_\true$ be the matrix of top $r$ singular vectors of $\X_0$.
By the triangle inequality,
\begin{align}
 \SE_2(\Upm^\sone_{T_{pm}}, \Ustar) & \le  \SE_2(\Ustar, \Upm_\true) + \SE_2(\Upm^\sone_{T_{pm}}, \Upm_\true)
\label{se_1}
\end{align}
Our goal is to bound the two terms in Eq.~\eqref{se_1}.}
\subsection{Proving Lemma~\ref{U0Ustar_bnd}: Bounding First Term in Eq.~\eqref{se_1}}\label{results_old}

The following is a corollary of  (i) the lemmas proved in \cite{lrpr_gdmin} and (ii) scalar consensus guarantee Proposition \ref{avgcons_prop_scalar} applied to show that each $\alpha^\sg$ is a close approximation of $\alpha$ (truncation threshold). We provide a proof in Appendix \ref{altdmin_results_init}.
\begin{corollary}\label{init_bnd_cor}
Let $\E[.]$ be $\E[. |\alpha^\sg, \forall  g \in [L]]$.
The following holds
\ben
\item $\E[\Xhat_0] = \Xstar \D = \Ustar \bSigma \Vstar \D$,
where
\[
\D :=diagonal(\beta_t(\alpha^\sg), \ k \in \S_g, \ g \in [L]), \text{ and }
\]
\[
\beta_{t}( s ) := \E[\zeta^2 \indic_{\{\|\xstar_{t}\|^2\zeta^2 \leq  s\}}].
\]
with $\zeta$  being a scalar standard Gaussian r.v.. Thus, $\E[\Xhat_0]$ is a rank $r$ matrix and
$
\E[\Xhat_0] \svdeq (\Ustar \Q) \check{\bSigma} \check{\Vstar},
$
where $\Q$ is an $r \times r$ unitary matrix, $\check{\bSigma}$ is an $r \times r$ diagonal matrix of its singular values and $\check{\Vstar}$ is an $r \times T$ matrix with orthonormal rows.

\item  Fix $0 < \eps_0 < 1$. Then, w.p. at least $1-\exp\left[(d+T)- c \frac{\eps_0^2 nT}{\kappa^4 \mu^2 r} \right] -   \exp(- c nT \epsilon_1^2 / \kappa^2 \mu^2)$,
\[
		\|\Xhat_{0} -\E[\Xhat_{0}]\| \le \eps_0 \sigmax
\]
\item Also, with the above probability,
\[
0.92 \le \min_t \beta_t(\alpha^\sg) \le \max_t \beta_t(\alpha^\sg) \le 1
\]
and so $\sigmamin(\D) \ge 0.92$ and $||\D||\le 1$.
\item With the above probability,
$
\sigma_r(\E[\X_0]) \ge 0.92 \sigmin.$
\een
\end{corollary}
Using Corollary \ref{init_bnd_cor} and Wedin's $\sin \theta$ theorem \cite{spectral_init_review} [Theorem 2.9], we can show the following.
\begin{lemma}
\label{U0Ustar_bnd}
Pick a $\tilde\delta_0 < 0.1$. If $nT \ge C  \kappa^2 \mu^2 (d + T)r / \tilde\delta_0^2$ and the event $\ev$ holds, then w.p. at least $1 - \exp(-c (d+T))$,
\[
\SE_2(\Ustar, \Upm_\true)  \le \tilde\delta_0.
\]
\end{lemma}

\begin{proof}
We apply Wedin's $\sin \theta$ theorem \cite{spectral_init_review} [Theorem 2.9] to $\X_0$ and $\E[\X_0]$ and use the fact that the span of top $r$ left singular vectors of $\E[\X_0]$ equals that of $\Ustar$ (see Corollary \ref{init_bnd_cor} above).
Applying Wedin and then using $\|\Ustar\|=\|\check{\Vstar}\|=1$, the bound on $\|\Xhat_{0} -\E[\Xhat_{0}]\|$ with $\eps_0 = 0.2 \tilde\delta_0 / \kappa$, and using the lower bound on $\sigma_r(\E[\X_0])$, we get
\begin{align*}
&\SE_2(\Upm_\true,\Ustar)
 \le \sqrt{2} \frac{ \|\X_0 - \E[\X_0]\|}{\sigma_r(\E[\X_0]) - 0 - \|\X_0 - \E[\X_0]\|}  \\
& \le 2 \frac{0.2 \tilde\delta_0 \sigmin}{0.92 \sigmin - 0.2 \tilde\delta_0 \sigmin} 
 \le 2 \frac{0.2 \tilde\delta_0 \sigmin}{0.92 \sigmin - 0.2 \cdot 0.1 \sigmin} \le \delta_0
\end{align*}
if $nT \gtrsim \kappa^4 \mu^2 (d+T) r / \tilde\delta_0^2$. We used $\tilde\delta_0 < 0.1$ to simplify the denominator.
\end{proof}

\subsection{Bounding Second Term in Eq.~\eqref{se_1}}
{\cblue To bound the second term of Eq.~\eqref{se_1}, we use the following noisy PM result from \cite{hardt2014noisy}.}
\begin{prop}[Noisy PM \cite{hardt2014noisy}, Corollary 1.1 with $p=r'$]\label{npm_res}
%
Let $\Z_\true$ denote the matrix of top $r$ singular vectors of a symmetric matrix $\bphi$ and let $\sigma_i$ denote it's $i$-th singular value. Consider the following approach 
\bi
\item For an  $r' \ge r$, let $\Z_0$ be an $d \times r'$ matrix with i.i.d. standard Gaussian entries. 
For $\tau=1$ to $T_{pm}$ do,
\bi
\item $\Z_\tau \leftarrow \mathrm{QR}(\tilde{\Z}_\tau)$ where $\tilde{\Z}_{\tau} \leftarrow \bphi \Z_{\tau-1} + \G_\tau $
\ei
\ei
Assume that for all $\tau$, $5\|\G_{\tau}\| \le \eps_{pm} (\sigma_r - \sigma_{r+1}) \mbox{~and~} 5\|\Z_\true^{\top}\G_{\tau}\| \le  (\sigma_r - \sigma_{r+1}) \dfrac{\sqrt{r'} - \sqrt{r-1}}{\gamma\sqrt{d}}$, for some fixed parameter $\gamma$ and $\eps_{pm} < 1/2$.
For any $T_{pm} > C \dfrac{\sigma_r}{\sigma_r - \sigma_{r+1}}\log (d\gamma/\eps_{pm})$ for a constant $C$,  with probability (w.p.) at least $ 1- \gamma^{-c_1(r'-r+1)}-e^{-c_2 d}$,
\[
\SE_2(\Z_{T_{pm}}, \Z_\true) < \eps_{pm}.
\]
\end{prop}
We apply Proposition \ref{npm_res} with $\gamma = d^{1/c_1}$, $r'=r$, 
\[
\bphi \equiv \X_0 \X_0^\top = \sum_g (\X_0)_\ssg (\X_0)_\ssg{}^\top,
\]
$\Z_\true = \Upm_\true$,
$\tilde\Z_\tau \equiv \tilde\Upm^\sone_\tau$, and $\Z_{\tau} \equiv \Upm_{\tau}^\sone$. Thus
\begin{align}
\G_\tau^\sone =  \tilde\Upm^\sone_\tau - \X_0 \X_0^\top \Upm_{\tau-1}^\sone.
\label{Gtaudef0}
\end{align}
Recall that $ \tilde\Upm_{\tau}^\sone = \mathrm{AvgCons}_1((\X_0)_\ssg  (\X_0)_\ssg{}^\top \Upm_{\tau-1}^\sg, g \in [L])$.

\begin{brem}
With picking $\gamma = d^{1/c_1}$ and $r'=r$, the result holds w.p. at least $1 - 1/d$. If we pick $r'=r+10$, the result will w.p. probability $1-d^{-10}$ and so on.
\end{brem}

Since $\|\Z_\true\|=1$, $\|\Z_\true^{\top}\G_{\tau}\| \le \|\G_\tau\|$. Also, $\sqrt{r} - \sqrt{r-1} \ge c/\sqrt{r}$ \cite{singvalsquare} [Theorem~1.1].
Using this, and the above choices of $\gamma,r'$, a sufficient condition to apply Prop~\ref{npm_res}  is
\[
5\|\G_\tau^\sone\| \le (\sigma_r - \sigma_{r+1})  \min\left( \eps_{pm} , \dfrac{c}{d^{1/c_1} \sqrt{r} } \right)
\]

%
We first lower bound $\sigma_r(\X_0 \X_0^\top) - \sigma_{r+1}(\X_0 \X_0^\top)$ and upper bound $\sigma_r(\X_0 \X_0^\top)$.
Using Corollary \ref{init_bnd_cor} with $\eps_0=0.1/\kappa$, 
if $nT \gtrsim \kappa^4 (d+T) r$, then,  whp,
\[
\|\X_0 - \E[\X_0]\| \le 0.1 \sigmin.
\]
Using this  and Weyl's inequality,
\begin{align}
& \sigma_r - \sigma_{r+1} \ge  0.8 \sigmin^2, \nn \\
& \sigma_r \le \|\X_0\|^2 \le 1.1 \sigmax^2
\label{sigma_bnds}
\end{align}
This follows using the above bound on $\|\X_0 - \E[\X_0]\|$, Weyl's inequality, and the facts that $\sigma_r(\X_0 \X_0^\top) = \sigma_r(\X_0)^2$,
$\sigma_r(\E[\X_0]) \ge 0.92 \sigmin$, $\sigma_{r+1}(\E[\X_0])=0$, and  $\sigma_r(\E[\X_0]) \le \|\E[\X_0]\| \le \sigmax^2$.
Using these, $(\sigma_r - \sigma_{r+1}) \ge (\sigma_r(\E[\X_0]) -  \|\X_0 - \E[\X_0]\|)^2 - \|\X_0 - \E[\X_0]\|^2 \ge (0.92 \sigmin - \eps_0 \sigmin)^2 - \eps_0^2\sigmin^2 \ge 0.92\sigmin (0.82 \sigmin) - 0.01\sigmin^2 \ge 0.6 \sigmin^2$.
%

Next we bound $\|\G_\tau^\sone\|$. Notice that it can be split as
\begin{align}
\G_\tau^\sone = \conserr^\sone + \Uerr^\sone  
\label{Gtaudef}
\end{align}
where
\begin{align*}
\conserr^\sone &: = \tilde\Upm_{\tau}^\sone - \sum_g (\X_0)_\ssg  (\X_0)_\ssg{}^\top \Upm_{\tau-1}^\sg \mathrm{~and~}\\
\Uerr^\sone &: = \sum_g (\X_0)_\ssg  (\X_0)_\ssg{}^\top \Upm_{\tau-1}^\sg -  \X_0 \X_0^\top \Upm_{\tau-1}^\sone\\
& = \sum_g (\X_0)_\ssg  (\X_0)_\ssg{}^\top (\Upm_{\tau-1}^\sg - \Upm_{\tau-1}^\sone).
\end{align*}
%
To bound $\conserr^\sone$, we need to define and bound the error of the input to AvgCons w.r.t. the desired sum, $\inperr^\sg$. Let
\[
\inperr^\sg := (\X_0)_\ssg (\X_0)_\ssg{}^\top \Upm_{\tau-1}^\sg -   \sum_{g'} (\X_0)_\ssgp  (\X_0)_\ssgp{}^\top \Upm_{\tau-1}^\sgp
\]
Using Claim \ref{norm_cols_removed},
\begin{align*}
& \| \inperr^\sg\|_F
 =\| \sum_{g' \neq  g} (\X_0)_\ssgp  (\X_0)_\ssgp{}^\top \Upm_{\tau-1}^\sgp \|_F \\
& \le \sum_{g'} \| (\X_0)_\ssgp  (\X_0)_\ssgp{}^\top \Upm_{\tau-1}^\sgp \|_F\\ 
& \le \max_{g'} \|\Upm_{\tau-1}^\sgp \|_F \sum_{g'} \| (\X_0)_\ssgp  (\X_0)_\ssgp{}^\top\|  \\
& \le \sqrt{r} \cdot  L \|\X_0\|^2  \le 1.1 L \sqrt{r} \sigmax^2
\end{align*}
if $nT \gtrsim \kappa^2 \mu^2 (d+T)r $.  The last inequality used Eq.~\eqref{sigma_bnds}.
Thus, using Lemma \ref{avgcons_prop} (consensus result Frob norm version), if
$T_\con \ge C \Tconexp \log (\frac{L}{\eps_\con})$ and $nT \gtrsim \kappa^2 \mu^2 (d+T)r $, then
\begin{align}
\|\conserr^\sone\|_F \le  1.1 L \sqrt{r} \eps_\con \sigmax^2.
\label{conserr_bnd}
\end{align}
%

Recall from our algorithm that, at each iteration $\tau$, we run a second consensus loop in which node-1 sends $\Upm_{\tau}^\sone$ while all other nodes send an $d \times r$ matrix of zeros. Thus, their true sum equals $\Upm_{\tau}^\sone$. 
Using Lemma \ref{avgcons_prop} (consensus result Frob norm version), we can show that this loop ensures that all nodes' estimates converge to within an $\eps_\con \cdot \|\Upm_{\tau}^\sone\|_F$ Frobenius norm distance of $\Upm_{\tau}^\sone$. 
To be precise, the input error (error between input of node to the consensus algorithm and the desired sum to be computed) for this consensus loop satisfies
\[
\inperr^\sone = \bm{0},  \ \inperr^\sg = \Upm_{\tau}^\sone - \bm{0}, \forall  g \neq 1
\]
Applying Lemma \ref{avgcons_prop} 
if $T_\con \ge C \Tconexp \log (L/\eps_\con )$, then
\[
\max_g \|\Upm_{\tau}^\sg - \Upm_{\tau}^\sone\|_F \le \eps_\con \max_g \|\inperr^\sg\| \le \eps_\con \sqrt{r}
\]
Thus, using the above bound and Eq.~\eqref{sigma_bnds},
\begin{align}
 \|\Uerr^\sone\|_F
& \le  \max_g \|\Upm_{\tau-1}^\sg - \Upm_{\tau-1}^\sone\|_F\sum_g \|(\X_0)_\ssg (\X_0)_\ssg{}^\top\|^2  \nn \\
& \hspace*{-10 mm} = \max_g \|\Upm_{\tau-1}^\sg - \Upm_{\tau-1}^\sone\|_F \cdot L \|\X_0\|^2 \nn \\
& \hspace*{-10 mm} \le   \sqrt{r} \eps_\con \cdot 1.1 L \sigmax^2 = 1.1 L \sqrt{r} \eps_\con \sigmax^2
\label{Uconserr_bnd}
\end{align}


\subsection{Proof of Theorem \ref{init_step_PM_claim}}
\begin{proof}[Proof of Theorem \ref{init_step_PM_claim}]
Applying noisy PM result and using Eqs.~\eqref{Gtaudef}, \eqref{conserr_bnd}, \eqref{Uconserr_bnd},
\begin{align}
\|\G_\tau^\sone\|_F   & \le 2.2 L \sqrt{r} \eps_\con  \sigmax^2
\label{Gtau_bnd1}
\end{align}

By Eq.~\eqref{sigma_bnds}, $\sigma_r(\X_0 \X_0^\top) - \sigma_{r+1}(\X_0 \X_0^\top) > 0.8 \sigmin^2$. Thus, in order to apply the noisy PM result to show that $\SE_2(\Upm^\sone_{T_{pm}}, \Upm_\true) \le \eps_{pm} =\delta_0/2$, we need
\[
\|\G_\tau^\sone\|< \min\left( 0.5 \delta_0, \frac{1}{d^{1/c_1} \sqrt{r}} \right) 0.8 \sigmin^2
\]
Using Eq.~\eqref{Gtau_bnd1}, and $\sigmin^2 = \sigmax^2/\kappa^2$, the above bound holds if
\[
\eps_\con \le \frac{1}{2.2 \kappa^2 \sqrt{r} } \cdot  \min\left( 0.5\delta_0, \frac{\sqrt{r}}{d} \right) 0.8
\]
%
To obtain the second bound of our Initialization theorem, we need $\eps_\con \sqrt{r} \le \rho_0$.
Both these requirements are satisfied if
\[
\eps_\con = 0.8 \min\left( \frac{\delta_0}{2 \kappa^2 L \sqrt{r}}, \frac{1}{\kappa^2 d L} , \frac{\rho_0}{\sqrt{r}}\right)
\]
This means that we need to set 
\[
T_\con = C \frac{( \log L + \log d + \log \kappa + \log (1/\delta_0) + \log(1/\rho_0) )}{{\log(1/\gamma(\W))}}
\]
To get the $T_{pm}$ expression, using Eq.~\eqref{sigma_bnds}, $\sigma_r(\X_0 \X_0^\top) - \sigma_{r+1}(\X_0 \X_0^\top) > 0.8 \sigmin^2$ and $\sigma_r(\X_0 \X_0^\top) \le 1.1 \sigmax^2$. Thus,
$\sigma_r / (\sigma_r - \sigma_{r+1}) < C \kappa^2$.
Thus, if
\[
T_{pm} \ge C \kappa^2 \log(\frac{d}{\delta_0}),
\]
$T_\con$ is set as above, and $nT \gtrsim \kappa^4 \mu^2 (d+T)r$,
then, by Proposition \ref{npm_res} (noisy PM result),
\[
\SE_2(\Upm^\sone_{T_{pm}}, \Upm_\true) <  0.5 \delta_0
\]
By Lemma \ref{U0Ustar_bnd}, $\SE_2(\Upm_\true, \Ustar) < 0.5 \delta_0 $ whp if $nT \gtrsim \kappa^2 \mu^2 (d+T)r/\delta_0^2$.
%
%
Thus, we have proved Theorem \ref{init_step_PM_claim}.
\end{proof}
%
\section{Simulation Experiments} \label{sims}
\newcommand{\prb}{p}
%
\begin{figure*}[h!]
\vspace*{-0.5in}
 \subcaptionbox{\footnotesize $L=20, p=0.5, T_\con=100$\label{fig:1}}{\includegraphics[width=2.3 in, height=3 in]{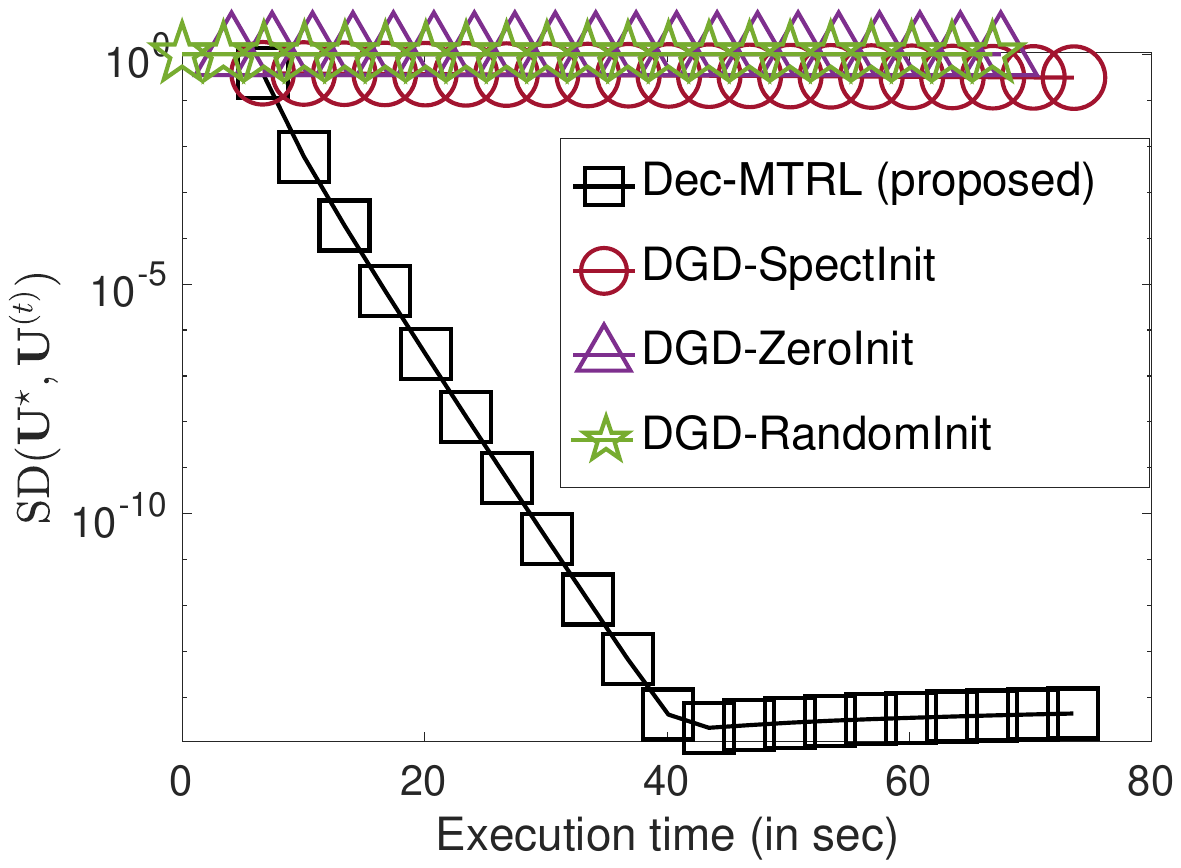}{\vspace*{-0.7 in}}}  \hspace{-1.3 em}%
 \centering
  \subcaptionbox{\footnotesize $T_\con$ varied, $p=0.5$\label{fig:3}}{\includegraphics[width=2.3in, height=3 in]{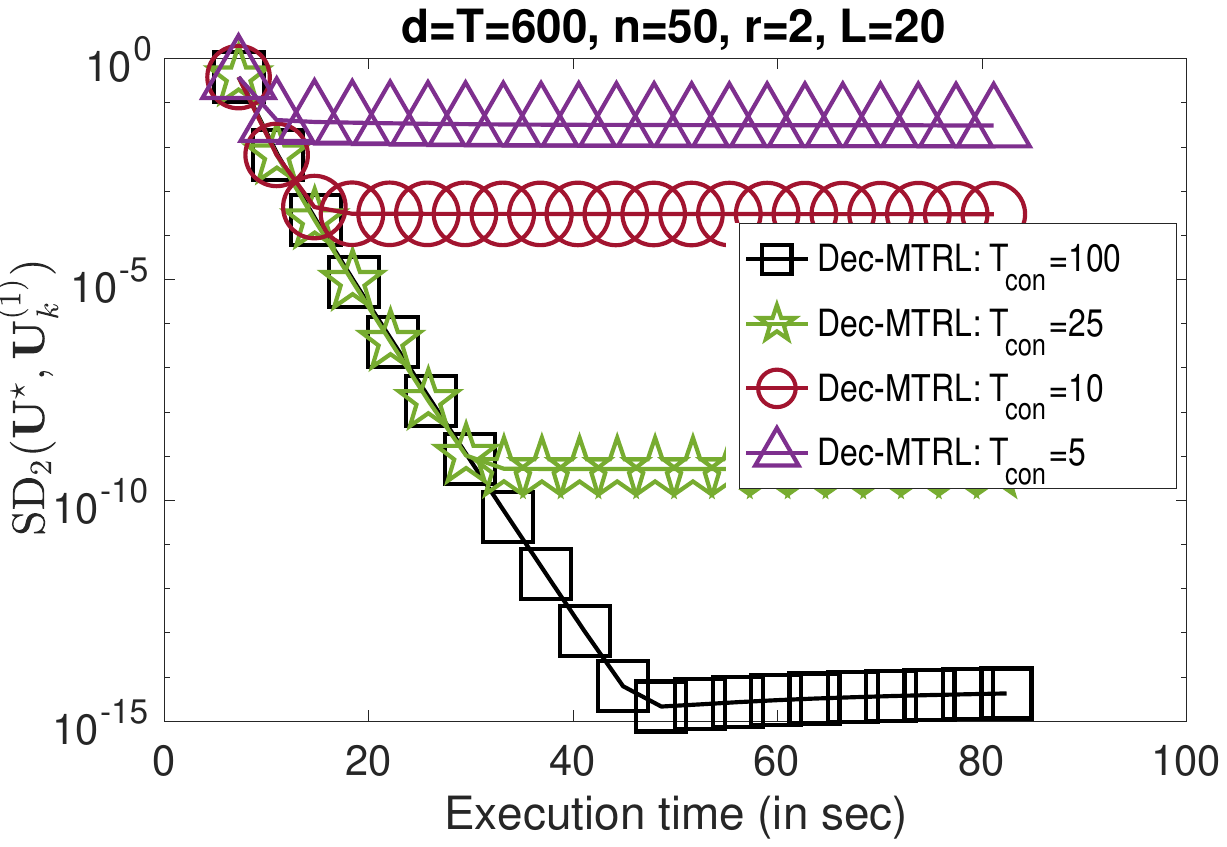}\vspace{-0.65 in}} 
 \subcaptionbox{\footnotesize $p$ varied, $T_\con=30$\label{fig:4}}{\includegraphics[width=2.3 in, height=3 in]{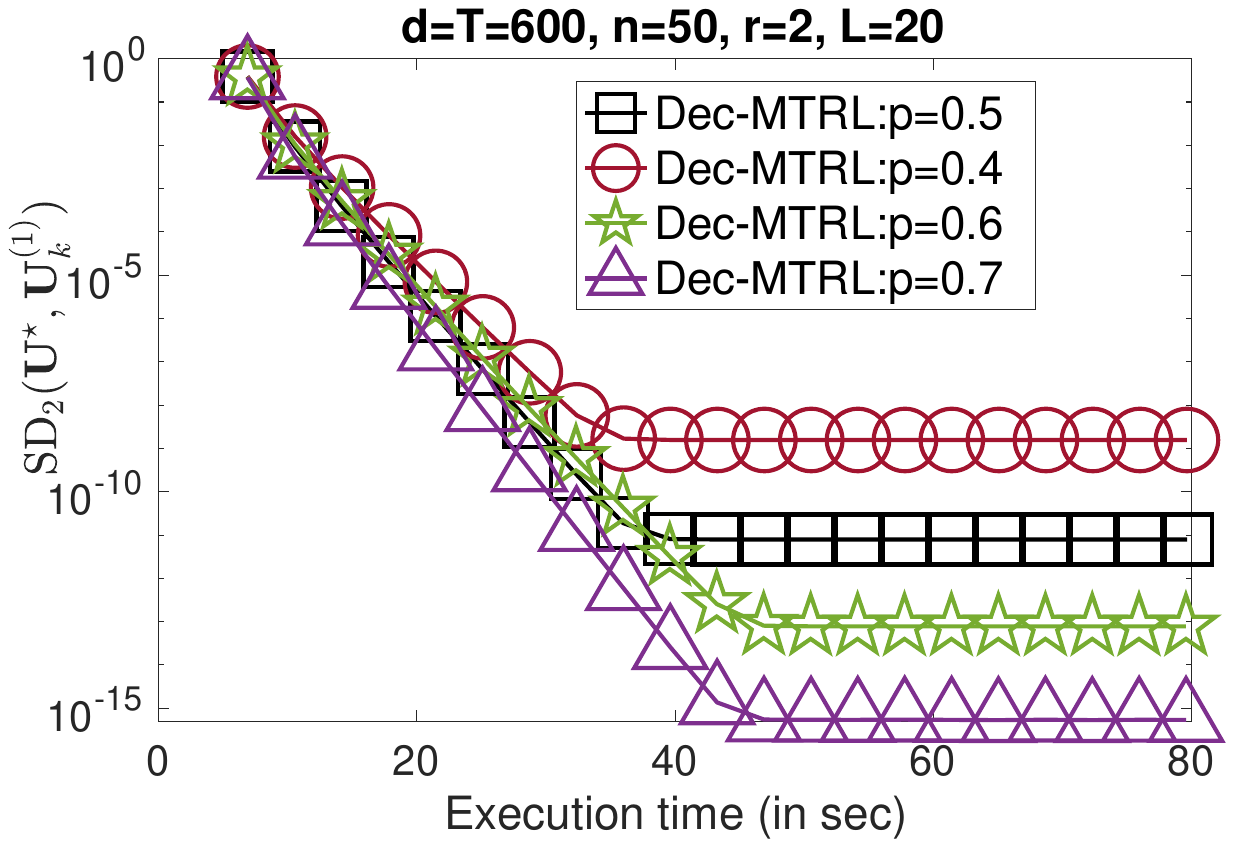} \vspace{-0.65 in}}
\hspace{-1.4em}%
\vspace{-2 mm}
\caption{\footnotesize {\em Error versus execution time in seconds.} In Fig.~\ref{fig:1}, we set $d=600$, $r=2$, $T=400$, $n=50$, $L=20$, and the network was generated as an ER graph with $p=0.5$ probability of nodes being connected.
We simulated until $T_\gd=400$ total GDmin iterations  In Figs.~\ref{fig:3} and~\ref{fig:4} $d=600$, $r=2$, $T=600$, $n=50$, $L=20$, and the network was generated as an ER graph with $p$ probability of nodes being connected.
In Figure~\ref{fig:3}, we compared the performance of our algorithm for $600$ tasks by varying the consensus iterations, $T_\con$.   In Fig.~\ref{fig:4}, we compared the performance of our algorithm for $600$ tasks by varying the edge probability $p$ of the network.  We varied $T_\con$ as $5, 10, 25$, and $100$ and  the probability of edge in the communication network $\calG$ as $0.4, 0.5, 0.6$ and $0.7$.}\label{fig:main3}
\vspace{-5 mm}
\end{figure*}
 \noindent{\em \bf Generating data and network.}  All the experiments were done using MATLAB on Dell precision tower 3420.
We simulated the network as an Erdos Renyi (ER) graph with $L$ vertices and with probability of an edge between any pair of nodes being  $\prb$.
Thus there exists an edge between any two nodes (vertices) $g$ and $j$ with probability $p$ independent of all other node pairs. The network was generated {\em once} outside our Monte Carlo loop. For a particular simulated graph, we used the {\em conncomp} function in MATLAB to verify that the graph is  connected.
%
We note that,  $\Xstar = \Ustar {\B}^{\star}$, where $\Ustar$ is an $d \times r$ orthonormal matrix. We generate the entries of $\Ustar$ by orthonormalizing an i.i.d standard Gaussian matrix. The entries of  ${\B}^{\star} \in \mathbb{R}^{r \times T}$ are generated from an i.i.d Gaussian distribution. The matrices $\A_t$s
 were i.i.d. standard Gaussian. We set the  step size of GD for the $g$-th agent as  $\eta^\sg = 0.4/n {\widehat{\sigma}}^{\star^2}_{\max}{}$, where ${\widehat{\sigma}}^{\star^2}_{\max}{}$ is obtained as the largest diagonal entry of ${\R}^\sg_{T_{pm}}$.
%
%
 In our experiments, we plot the subspace error $\SE_2(\U_\t^\sone,\Ustar)$ at each iteration $\t$ on the y-axis and the execution time taken by the algorithm until GD iteration $k$ on the x-axis. Averaging is over 100 trials.


 \noindent{\em \bf Comparison with existing decentralized algorithms.}  
We compared the performance of our  Dec-MTRL algorithm with existing decentralized algorithms. 
As noted earlier there are no existing algorithms for solving MTRL in a decentralized federated setting. 
A potential benchmark to compare with is by modifying the projected GD step for $\U$ by the decentralized GD approach \cite{nedic2009distributed, yuan2016convergence}: $\U_+^\sg \leftarrow  QR( \frac{1}{\dg} \sum_{g' \in \N_g} \U^\sgp  - \eta \nabla_\U f_g (\U^\sg, \B^\ssg) )$. We note that both of these works consider a convex case, unlike our problem. To initialize, we can either use the zero initialization as in \cite{yuan2016convergence} or a random initialization as in \cite{nedic2009distributed} or use our spectral initialization, Algorithm~\ref{altgdmin_dec_init}, which guarantees a good initialization.
In Figure~\ref{fig:1} we set $d=600, n=50, r=2, L=20$ and $p=0.5$.
  In Figure~\ref{fig:1}, we compared Dec-MTRL with the decentralized algorithms obtained by modifying the GD step using the approaches in \cite{yuan2016convergence}.  In Figure~\ref{fig:1} DGD-RandomInit refers to the setting with random initialization, DGD-ZeroInit refers to the setting with zero initialization), and DGD-SpectInit refers to the setting described initialized via our proposed spectral initialization.
  From the plots, we observe that the  our algorithm converges while DGD approaches fail.
   The key reason why the approach of \cite{nedic2009distributed, yuan2016convergence} and the follow-up works  designed for standard GD  cannot be used for updating $\U$ is because it involves averaging the partial estimates $\U^\sg$,  $g \in [L]$, obtained locally at the different nodes. However, since $\U^\sg$'s are subspace basis matrices, their numerical average will not provide a valid ``subspace mean''. Moreover, the 
cost functions considered in \cite{nedic2009distributed, yuan2016convergence} are convex, while ours is non-convex.


\noindent{\em \bf Comparison by varying $T_\con$ and $p$.} We tested the performance of our Dec-MTRL algorithm for different values of consensus iterations, $T_\con$, and edge probability of the communication network $\calG$. The plots are presented in Figure~\ref{fig:main3}. The plots of this experiment are given in Figure~\ref{fig:3}.  We set the parameters as  $d=T=600, n=50, r=2, L=20$  and $T_\gd=400$. 
 We plot the matrix estimation error (at the end of the iteration) $\SE(\Ustar, \U_\t^{(1)})$ and the execution time-taken (until the end of that iteration) on the y-axis and x-axis, respectively. 
  In Figure~\ref{fig:3}, we set $p=0.5$ and varied the consensus iterations as, $T_\con=5, 10, 25,$ and $ 100$. The experimental results show that subspace error decreases while increasing the $T_\con$ and for $T_\con=100$ above the error value saturates.
In Fig.~\ref{fig:4}, we set $T_{\con}=30$ and we varied the values of the edge probability as, $p=0.4, 0.5,0.6,$ and $0.7$.  
 From Fig.~\ref{fig:4} we observe that as expected the subspace error $\SE(\Ustar, \U_\t^{(1)})$ decreases when the connectivity of the network improves which is expected. 
 
{\cblue\noindent{\em \bf Comparison with centralized algorithms:} We compared the performance of the proposed Dec-MTRL algorithm with that of three other existing algorithms that are widely studied in the LR literature, AltGDmin \cite{lrpr_gdmin}, Altmin \cite{lowrank_altmin}, AltGD \cite{rpca_gd}, and ProjGD \cite{fastmc}.
The plots are presented in Figure~\ref{fig:main2}.
It is worth noting that in this comparison, Dec-MTRL is the only decentralized algorithm, while the rest fall under the category of centralized algorithms. 
We set $d=T=600$, $r=4$, the probability of network connectivity $p=0.5$, and the number of agents is set to $L=20$.  We chose different values of $n$ in our experiments to evaluate the performance of the different algorithms based on the extend of compression the data undergoes.
In Figures~\ref{fig:ac2} and \ref{fig:ac4},  we  varied the value of $n$ as $50, 30$, respectively. From Figures~\ref{fig:ac2} and \ref{fig:ac4} we see that as the rate of compression increases (i.e., $n/d$ decreases), the ProjGD, AltGD algorithms no longer converge. The AltGDmin and Dec-MTRL still converge for these highly compressed cases, thereby validating the sample efficiency of these approaches over the other three. 
Additionally, we notice from the plots that the Altmin algorithm is computationally intensive.
Another observation is that the error decay is not monotonic in the figures. This is because we use the same value of consensus iterations in our experiments.  
 Although the consensus error will be below $\eps_\con$, the exact error value in the different GD iterations will be different.  
 During certain iterations of GD,  while computing the gradient through average consensus, the consensus error may be lower than in previous steps.
 }

\begin{figure*}[h!]
\vspace*{-0.6 in}
\centering
\subcaptionbox{\footnotesize  Error-X vs. Execution time: $n=50$\label{fig:ac2}}{\includegraphics[width=3.0in, height=3.4 in]{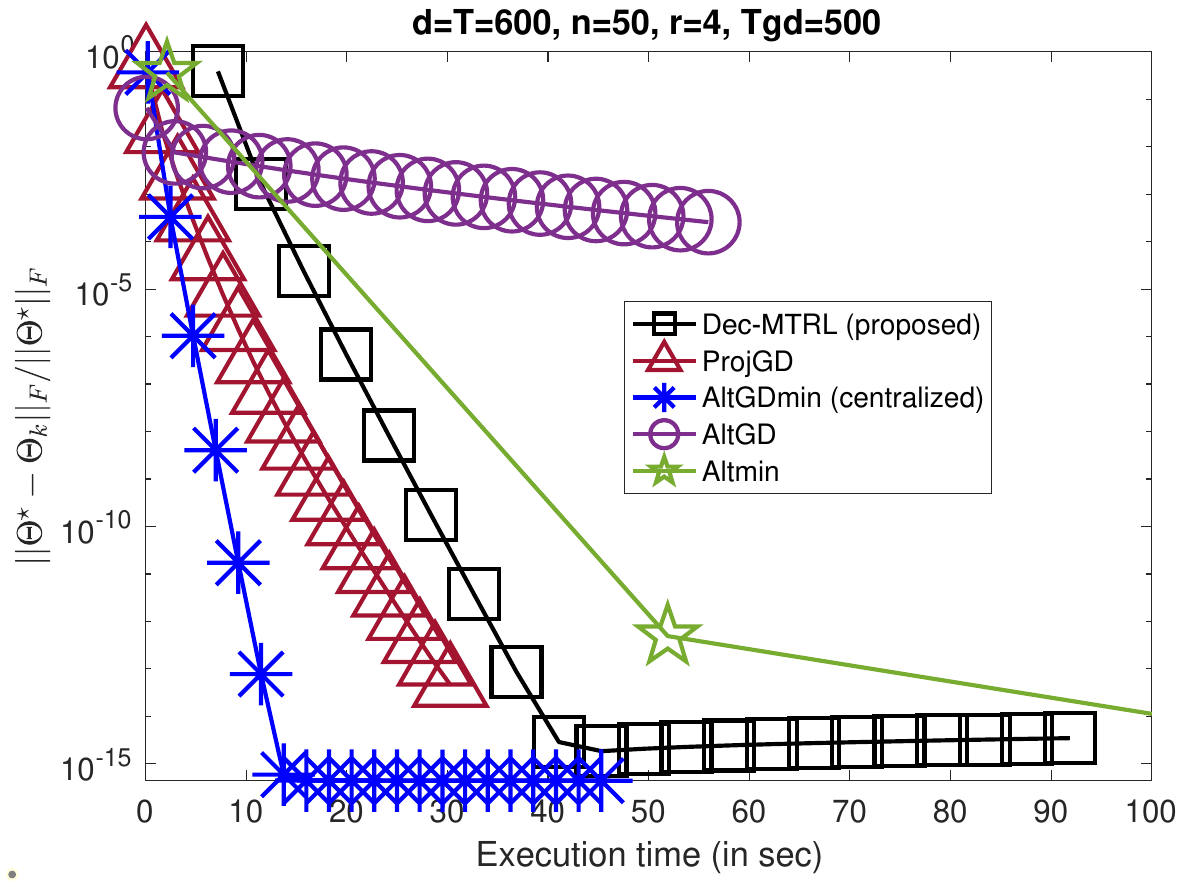}\vspace{-0.85 in}}\hspace{0.3 em}
 \subcaptionbox{\footnotesize  Error-X vs. Execution time: $n=30$\label{fig:ac4}}{\includegraphics[width=3.0in, height=3.4 in]{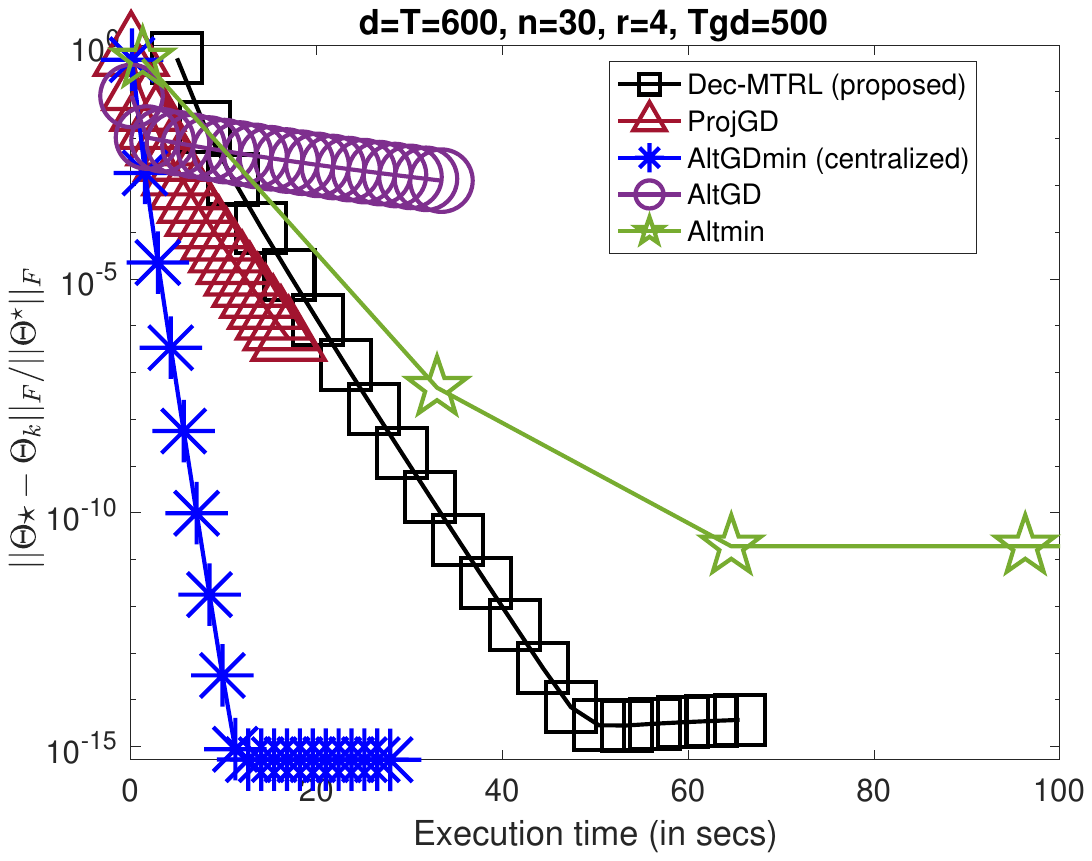}\vspace{-0.85 in}}
\hspace{-1.4em}%
\caption{\footnotesize
{\em Error versus execution time plot with time in seconds.}
In both figures $d=T=600$, $r=4$, $L=20$,  and the network was generated as an ER graph with $p=0.5$ probability of nodes being connected. We compared the performance of our proposed algorithm (Dec-MTRL)  with the existing centralized counterparts from \cite{lrpr_gdmin} and with the LR matrix recovery algorithms,  Altmin \cite{lowrank_altmin}, ProjGD \cite{fastmc}, and AltGD \cite{rpca_gd} algorithms. The network was generated as an ER graph with $p=0.5$ probability.
 In Figures~\ref{fig:ac2}, \ref{fig:ac4}, we present the plots for error vs. execution time  for $n=50, 30$, respectively. 
 For improved visualization, we have set the x-limit to $100$ seconds.
 This choice aligns with the runtime of all algorithms except Altmin, which exceeds $1000$ seconds.}\label{fig:main2}
\vspace{-4mm}
\end{figure*}
 \section{Conclusions} \label{conclude}
In this work, we developed and analyzed an alternating projected gradient descent and minimization algorithm for decentralized multi-task representation learning. 
We presented the convergence guarantees and sample complexity results for the proposed approach.
Open questions for future work include (i) how to obtain a guarantee for which the number of consensus iterations needed in each GD step iteration does not depend on the final desired accuracy level, and (ii) how to improve the initialization guarantee.
%
\bibliographystyle{IEEEtran}
\bibliography{Decentralized-ST, nv-refs, tipnewpfmt_kfcsfullpap, Bandits}
\textbf{Shana Moothedath} (Member, IEEE) received her Ph.D. degree in Electrical Engineering from the Indian Institute of Technology Bombay (IITB) in 2018, and she was a postdoctoral scholar in Electrical and Computer Engineering at the University of Washington, Seattle till 2021. Currently, she is Harpole-Pentair Assistant Professor of Electrical and Computer Engineering at Iowa State University. Her research focuses on distributed decision-making, control and security of dynamical systems, and  signal processing. She received the NSF CAREER Award in 2025, the Best Research Thesis Award at IITB in 2019, and selected as a MIT-EECS Rising Star in 2019.\par
\noindent\textbf{Namrata Vaswani} (Fellow, IEEE) received the B.Tech. degree from Indian
Institute of Technology (IIT Delhi), India, in 1999, and the Ph.D. degree
from the University of Maryland, College Park, in 2004. She is currently a
Professor of Electrical and Computer Engineering and an Anderlik Professor
of engineering at Iowa State University. 
Her research interests include statistical machine learning
and signal processing and their applications in medical imaging and video.
She was a fellow of AAAS in 2023. She was a recipient of the IEEE
Signal Processing Society Best Paper Award (2014), the University of
Maryland ECE Distinguished Alumni Award (2019), and the Iowa State MidCareer Achievement in Research Award (2019). She served as Area Editor
for IEEE Signal Processing Magazine and Associate Editor for IEEE
TRANSACTIONS ON INFORMATION THEORY and IEEE TRANSACTIONS ON
SIGNAL PROCESSING, and has guest-edited a special issue for PROCEEDINGS
OF THE IEEE.
 
\appendices
\renewcommand\thetheorem{\Alph{section}.\arabic{theorem}}

\clearpage
\section{Preliminaries}

\begin{prop}[Theorem~2.8.1, \cite{versh_book}] \label{proposition1}
Let $X_1, \cdots, X_N$ be independent, mean zero, sub-exponential random variables. Then, for every $\t \geqslant 0$, we have
$$
\scalemath{0.9}{\mathbb{P} \Bigl\{ |\sum_{i=1}^N X_i \geqslant \t| \Bigl\} \leqslant 2 \exp \left[-c \min \left( \frac{\t^2}{\sum_{i=1}^N \norm{X_i}_{\psi_1}^2}, \frac{\t}{\max_i \norm{X_i}_{\psi_1}^2} \right) \right], }
$$
where $c > 0$ is an absolute constant. 
\end{prop}

\section{Results for initialization taken from \cite{lrpr_gdmin}}
\label{altdmin_results_init}
The following linear algebra results are used in the proofs below and also in the main paper. 
\begin{claim}
\label{linalg}
If $\X = \U \M$ with $\U$ being $d \times r$ and orthonormal, and $\M$ being $r \times T$,
$
\sigma_r(\X)=  \sigmamin(\M).
$
If $\A = \B \C$ and $\A$ is tall (or square) and  both $\B, \C$ are tall (or square), then
\[
\sigmamin(\A)=   \min_{\z \neq 0} \frac{\|\A \z\|}{\|\z\|} \ge  \sigmamin(\B) \sigmamin(\C)
\]
\end{claim}

\begin{lemma}[\cite{lrpr_gdmin}] \label{EX0}
Conditioned on $\alpha^\sg$s, we have the following conclusions.
\ben
\item $\E[\Xhat_0] = \Xstar \D$ 
where $\D :=diagonal(\beta_t(\alpha^\sg), \ t \in \S_g, \ g \in [L])$
with $\beta_t(.)$ is as defined earlier in Corollary \ref{init_bnd_cor}. Recall from notation that $\Xstar \svdeq \Ustar \bSigma \Vstar$. 
\item Thus, $\E[\Xhat_0]$ is a rank $r$ matrix and
\[
\E[\Xhat_0] \svdeq (\Ustar \Q) \check{\bSigma} \check{\Vstar}
\]
where $\Q$ is an $r \times r$ unitary matrix, $\check{\bSigma}$ is an $r \times r$ diagonal matrix of its singular values and $\check{\Vstar}$ is an $r \times T$ matrix with orthonormal rows. This means that 
the span of top $r$ left singular vectors of $\E[\Xhat_0]$ equals the column span of $\Ustar$. 
\item Also, 
$\sigma_r(\E[\X_0])= \sigmamin(\bSigma \Vstar \D) =\sigmamin(\D \Vstar^\top \bSigma) \ge \sigmamin(\D) \sigmamin( \Vstar^\top \bSigma) \ge \sigmamin(\D) \sigmamin( \Vstar^\top) \sigmamin(\bSigma) = \min_j |D_{jj}| \cdot 1 \cdot \sigmin
$
\een 

%
\end{lemma}

The following lemma is an easy consequence of the scalar consensus result and the sub-exponential Bernstein inequality. %
\begin{lemma}[$\alpha^\sg$ consensus]
\label{alpha_bnds}
Let  $T_\con = C \Tconexp \log(L/\eps_\con)$. Recall that $\alpha = 9 \kappa^2 \mu^2 \sum_g \sum_{t \in \S_g} \|\y_t\|^2  / (nT)$ and $\alpha^\sg$ is its consensus estimate after $T_\con$ iterations.
The following hold:
\ben
\item w.p. at least $1 - \exp(- c nT \epsilon_1^2 / \kappa^2 \mu^2)$,
\[
|\alpha - 9 \kappa^2 \mu^2\frac{\|\Xstar\|_F^2}{T} | \le \eps_1 9 \kappa^2 \mu^2\frac{\|\Xstar\|_F^2}{T}
\]

\item  $\max_{g \in [L]}|{\alpha}^\sg - \alpha|   \le  \eps_\con \alpha$

\item Thus, setting $\eps_\con = 0.01$ and $\eps_1 = 0.01$,
\\ w.p. at least $1 - \exp( - c nT  / \kappa^2 \mu^2)$, for all $g \in [L]$,
{\small
\[
\ev^\sg: =  \left\{  8.98 \kappa^2 \mu^2 \frac{\|\Xstar\|_F^2}{T} \le \alpha^\sg \le 9.02 \kappa^2 \mu^2 \frac{\|\Xstar\|_F^2}{T}  \right\}\mbox{~holds.}
\]
} 
\een
%
\end{lemma}

\begin{proof}
The first claim is the same as Fact 3.7 of \cite{lrpr_gdmin}. It is a direct consequence of the sub-exponential Bernstein inequality \cite{versh_book}.
The second claim follows by applying Lemma \ref{avgcons_prop} as follows.
Let $C_1 = 9 \kappa^2 \mu^2 / nT$.
The desired sum is $C_1 \sum_g \sum_{t \in \S_g} \|\y_t\|^2$, and the input at node $g$ is $C_1 \sum_{t \in \S_g} \|\y_t\|^2$. Thus the input error at node $g$, $\inperr^\sg = C_1 \sum_{g' \neq g} \sum_{t \in \S_{g'}} \|\y_t\|^2 \le C_1 \sum_{g'} \sum_{t \in \S_{g'}} \|\y_t\|^2 = \alpha$. The inequality follows since all summands are non-negative. Applying Lemma $\ref{avgcons_prop}$, then $\conserr^\sg = |\alpha - \alpha^\sg| \le \eps_\con \max_g |\inperr^\sg| \le \eps_\con \alpha$.
The third claim is an easy consequence of the first two.
\end{proof}

\begin{fact}[\cite{lrpr_gdmin}]
\label{betak_bnd}
		$
 \min_t  \E\left[\zeta^2 \indic_{ \left\{ |\zeta| \leq   \frac{ \sqrt{8.98} \kappa \mu \|\Xstar\|_F }{ \sqrt{q}\|\xstar_{t}\| } \right\} } \right] \geq 0.9
	$
	\end{fact}
Combining the above two results, w.p. at least $1 - \exp( - c nT  / \kappa^2 \mu^2)$,
\begin{align}
\min_g \min_{t \in \S_g} \beta_t (\alpha^\sg) \ge  \min_t  \E\left[\zeta^2 \indic_{ \left\{ |\zeta| \leq   \frac{ \sqrt{8.98} \kappa \mu  \|\Xstar\|_F }{ \sqrt{T}\|\xstar_{t}\| } \right\} } \right]  \ge 0.9
\label{minbetak}
\end{align}

The following lemma can be proved exactly as done in \cite{lrpr_gdmin}. It uses the last item of Lemma \ref{alpha_bnds} in its proof. The only difference in the proof now is that the thresholds $\alpha^\sg$ are different for subsets of columns. But, since we condition on $\alpha^\sg$s with $\alpha^\sg \in \ev^\sg$, nothing changes.
\begin{lemma}[\cite{lrpr_gdmin}]  \label{init_bnd}
Fix $0 < \eps_1 < 1$. Then,  w.p. at least $1-\exp ( (d+T)-c\epsilon_1^2 nT/\mu^2\kappa^2 )s$, conditioned on $\alpha^\sg$ s, for  $\alpha^\sg \in \ev^\sg$,
		\[
		\|\Xhat_{0} -\E[\Xhat_{0}]\| \leq 1.1 \eps_1 \|\Xstar\|_F
		\]
\end{lemma}
By setting $\eps_1 = \eps_0/ 1.1 \sqrt{r} \kappa$ in the above lemma, and using Lemma \ref{alpha_bnds} (last claim) to remove the conditioning on $\alpha^\sg$s, we can conclude that
\\ $\|\Xhat_{0} -\E[\Xhat_{0}]\| \leq \eps_0 \sigmin$ w.p. at least
\\ $1 - \exp ( (d+T)-c\epsilon_1^2nT/\mu^2\kappa^2 ) - \exp(\log L - c nT  / \kappa^2 \mu^2)$.

This fact, Lemma \ref{EX0}, and  Eq.~\eqref{minbetak} constitute  Corollary \ref{init_bnd_cor} in the main paper.

\section{Proof of Lemma \ref{grad_bnd_new2}}\label{grad_bnd_new2_proof}
\begin{proof}
Assume everything given  in Theorem~\ref{min_step_claim} holds.
 Using bounds on $\|\B\|$ and  $\|\Xstar - \X\|_F$ from  Theorem~\ref{min_step_claim}, if $\delta_\t < \frac{c}{\sqrt{r} \kappa}$,
\begin{align}
\| \E[\gradU]\|   & =  \| \sum_t n (\x_t-{\xstar_t})\b_t{}^\top \|  \nn\\
&\le n\|\X - \Xstar\| \cdot \| \B\|  \nn\\
&  \le n\| \X -\Xstar\|_F \cdot \| \B\|  \le 1.1 n\delta_\t \sqrt{r}  \sigmax^2. \nn
\end{align}

Next, we bound
$\| \gradU-\E[\gradU]\|
 = \max_{\|\w\|=1, \|\z\|=1}  \w^\top ( \sum_t \sum_i \a_\ik \a_\ik^{\top}(\x_t-\xstar_t) \b_t^\top  - \E[\cdot]) \z
 $
We bound the above for fixed unit norm $\w,\z$ using sub-exponential Bernstein inequality (Theorem 2.8.1 of \cite{versh_book}) presented in Proposition~\ref{proposition1}, followed by a standard epsilon-net argument.
Consider the following for fixed unit norm $\w,\z$
\[
\sum_t \sum_i (\w^\top \a_\ik)  (\b_t^\top \z) \a_\ik^{\top}(\x_t-\xstar_t)- \E[\cdot])
\]
Observe that the summands are independent, zero mean, sub-exponential r.v.s with sub-exponential norm $K_\ik \le C \| \w \| |\z^\top \b_t| \| \x_t-\xstar_t\| =C |\z^\top \b_t| \| \x_t-\xstar_t\| $.
We apply the sub-exponential Bernstein inequality, Proposition~\ref{proposition1}, with $k = \eps_1 \delta_\t n  \sigmin^2$. We have
\begin{align*}
\frac{k^2}{\sum_\ik K_\ik^2}  
& \ge \frac{\eps_1^2\delta_\t^2 n^2  \sigmin^4}{ n  \sum_t \|\x_t -\xstar_t\|^2 (\z^\top \b_t)^2   }\\
& \ge \frac{\eps_1^2\delta_\t^2 n^2  \sigmin^4}{ n \max_t \|\x_t-\xstar_t\|^2 \sum_{t} (\z^\top \b_t)^2 }\\
& = \frac{\eps_1^2\delta_\t^2 n^2  \sigmin^4}{n \max_t \|\x_t-\xstar_t\|^2 \|\z^\top \B\|^2 }\\
& \ge \frac{\eps_1^2\delta_\t^2 n^2  \sigmin^4 T }{ 1.4^2 n \mu^2 \delta_\t^2 r \sigmax^2 \|\B\|^2 }\\
& \ge \frac{\eps_1^2\delta_\t^2 n^2  \sigmin^4 T }{ 1.4^2 n \mu^2 \delta_\t^2 r \sigmax^2 1.1 \sigmax^2 } =  c \frac{\eps_1^2 n  T }{ \kappa^4 \mu^2  r }\\
\frac{\t}{\max_\ik K_\ik} & \ge  \frac{\eps_1\delta_\t n  \sigmin^2 }{\max_t \| \b_t\| \max_t \|\x_t - \xstar_t\| }
\ge c \frac{\eps_1 n  T }{\kappa^2 \mu^2  r}
\end{align*}
In the above, we used
(i) $\sum_t (\z^\top \b_t)^2 = \|\z^\top \B\|^2 \le \|\B\|^2$  since $\z$ is unit norm,
(ii) Theorem \ref{min_step_claim} to bound $ \|\B\| \le 1.1 \sigmax$, and
(iii) Theorem \ref{min_step_claim} followed by Assumption \ref{right_incoh2} (right incoherence) to bound $\|\x_t-\xstar_t\| \le \delta_\t \cdot \mu \sigmax \sqrt{r/T}$ and $|\z^\top \b_t| \le \|\b_t\| \le 1.1 \|\bstar_t\| \le 1.1 \mu \sigmax \sqrt{r/T}$.

For $\eps_1 < 1$, the first term above is smaller (since $1/\kappa^4 \le 1/\kappa^2$), i.e., $\min(\frac{\t^2}{\sum_\ik K_\ik^2},\frac{\t}{\max_\ik K_\ik} ) =  c\frac{\eps_1^2 n T }{\kappa^4 \mu^2 r}.$
 Thus, by sub-exponential Bernstein,  w.p. at least  $1-\exp (- c\frac{\eps_1^2 n T }{\kappa^4 \mu^2 r} )$ 
\[
	 \w^\top(\gradU - \E[\gradU] ) \z \le  \eps_1 \delta_\t n \sigmin^2
\]
Using a standard epsilon-net argument to bound the maximum of the above over all unit norm $\w,\z$, we can conclude that
\[
	 \|\gradU - \E[\gradU] \| \le  1.1 \eps_1 \delta_\t n \sigmin^2
\]
w.p. at least  $1-  \exp ( C(d+r) -c\frac{\eps_1^2 n T }{\kappa^4 \mu^2 r} )$.
The factor of $ \exp(d+r)$ is due to the epsilon-net over $\w$ which is an $d$-length unit norm vector (and thus the epsilon net covering it is of size $(1+ 2/\eps_{net})^d = C^d$ with $\eps_{net}=c$) and the epsilon-net over $\z$ (and thus the epsilon net covering it is of size $C^r$). Union bound over both thus gives a factor of $C^{dr}$.  Since Theorem~\ref{min_step_claim} holds w.p. $1-\exp(\log T+r-cn)$, by union bound all the above claims hold w.p.  $1-\exp ( C(d+r) -c\frac{\eps_1^2 n T }{\kappa^4 \mu^2 r} )-\exp(\log T+r-cn)$.

By replacing $\eps_1$ by $\eps_1/1.1$, our bound becomes simpler (and $1/1.1^2$ gets incorporated into the factor $c$ in the exponent). 
We have thus proved Lemma \ref{grad_bnd_new2}.
 \end{proof}
\end{document}